\documentclass[12pt]{article}

\usepackage{lmodern} %
\usepackage{graphicx}
\usepackage{xcolor}
\usepackage[margin=1.25in, top=1.5in, bottom=1.5in]{geometry}
\usepackage{amsmath}
\usepackage{amssymb}
\usepackage{dsfont} %
\usepackage{xparse} %
\usepackage{math}   %
\usepackage{amsthm} %
\usepackage{enumitem} %
\usepackage{appendix}
\usepackage{mathrsfs} %
\usepackage[hidelinks]{hyperref} %
\hypersetup{
  pdftitle={Cross-Fitting Under Nonregularity: Normality and Inference via Locality},
  pdfauthor={Bruno Fava},
  pdfsubject={}
}
\usepackage{cleveref} %
\usepackage{booktabs} %
\usepackage{array} %
\usepackage{threeparttable} %
\usepackage[bottom]{footmisc} %

\newtheorem{lemma}{Lemma}

\newtheorem{theoreminner}{Theorem}

\newenvironment{theorem}[1][]{
  \begin{theoreminner}[#1]
}{
  \end{theoreminner}
}

\newtheorem{propositioninner}{Proposition}

\newtheorem{exampleinner}{Example}

\newenvironment{example}[1][]{
  \begin{exampleinner}[#1]
}{
  \end{exampleinner}
}

\NewDocumentEnvironment{RevExample}{m}{
    \begingroup
    \renewcommand{\theexampleinner}{\ref{#1}}%
    \begin{example}[Revisited]%
}{%
    \end{example}%
    \endgroup
}

\newtheorem{assumptioninner}{Assumption}

\newenvironment{assumption}[1][]{
  \begin{assumptioninner}[#1]
}{
  \end{assumptioninner}
}

\newtheorem{remarkinner}{Remark}

\newenvironment{remark}[1][]{
  \begin{remarkinner}[#1]
}{
  \end{remarkinner}
}

\newcommand{\appendixnumbering}{%
  \counterwithin{equation}{section}%
  \counterwithin{theoreminner}{section}%
  \counterwithin{propositioninner}{section}%
  \counterwithin{exampleinner}{section}%
  \counterwithin{lemma}{section}%
  \counterwithin{corollary}{section}%
  \counterwithin{remarkinner}{section}%
  \counterwithin{definition}{section}%
  \counterwithin{assumptioninner}{section}%
}

\crefname{equation}{}{}
\Crefname{equation}{}{}
\crefname{theoreminner}{Theorem}{Theorems}
\Crefname{theoreminner}{Theorem}{Theorems}
\crefname{propositioninner}{Proposition}{Propositions}
\Crefname{propositioninner}{Proposition}{Propositions}
\crefname{exampleinner}{Example}{Examples}
\Crefname{exampleinner}{Example}{Examples}
\crefname{assumptioninner}{Assumption}{Assumptions}
\Crefname{assumptioninner}{Assumption}{Assumptions}
\crefname{remarkinner}{Remark}{Remarks}
\Crefname{remarkinner}{Remark}{Remarks}
\crefname{appendix}{Appendix}{Appendices}
\Crefname{appendix}{Appendix}{Appendices}

\newlist{assumptionenum}{enumerate}{1}
\setlist[assumptionenum]{label=\upshape(\roman*), ref=\theassumptioninner(\roman*)}
\crefname{assumptionenumi}{Assumption}{Assumptions}
\Crefname{assumptionenumi}{Assumption}{Assumptions}

\newcommand{\rmd}{\mathrm{d}}

\newcommand{\etahat}{{\hat{\eta}}}

\newcommand{\thetahat}{\hat{\theta}}
\newcommand{\betahat}{{\hat{\beta}}}

\newcommand{\muhat}{{\hat{\mu}}}
\newcommand{\nuhat}{{\hat{\nu}}}
\newcommand{\sigmahat}{\hat{\sigma}}

\newcommand{\CIalpha}{\widehat{\mathrm{CI}}_{\alpha}}

\newcommand{\s}{\mathsf{s}}
\newcommand{\stilde}{{\tilde{\mathsf{s}}}}

\usepackage{natbib}
\title{\vspace{-65pt}Cross-Fitting Under Nonregularity:\\Normality and Inference via Locality\vspace{1.5em}}
\author{\Large Bruno Fava\thanks{Contact: bfava@stanford.edu. I am really grateful to Stefan Wager for numerous very helpful conversations. 
This research was supported in part through the computational resources and staff contributions provided for the Quest high performance computing facility at Northwestern University which is jointly supported by the Office of the Provost, the Office for Research, and Northwestern University Information Technology.}\\[0.5em]
\large Stanford University}
\date{\vspace{1.2em}\today}

\begin{document}

\maketitle

\begin{abstract}
Cross-fitting is routine in much of applied research. %
While conventional confidence intervals that ignore cross-fold dependence are asymptotically valid in several settings, they undercover in many applications that share a common form of nonregularity: from the classic cross-validation problem of testing whether a fitted model outperforms another, to testing for heterogeneous treatment effects with machine learning, to estimating the value of a potentially non-unique optimal treatment regime. 
Exploiting a new locality condition, I show that a large class of cross-fitting estimators still satisfies a central limit theorem despite the nonregularity, but with an asymptotic variance that must be adjusted for the cross-fold correlation. 
Then, I propose a method for estimating this correlation and construct new confidence intervals that attain asymptotically nominal coverage. 
Finally, I show that the proposed confidence intervals attain approximately nominal coverage in a simulation study with random forests and neural networks. 
\end{abstract}

\clearpage

\section{Introduction}
\label{sec.intro}

Cross-fitting is widely popular in applied research. 
By using every observation both for training and for testing, it improves efficiency over sample-splitting, but doing so introduces statistical dependence across folds that can make inference challenging. 
In many settings, including semiparametric estimation and debiased machine learning, this dependence is asymptotically negligible, and confidence intervals that treat folds as independent have approximately nominal coverage \citep{schickAsymptoticallyEfficientEstimation1986,hubbardStatisticalInferenceData2016b,chernozhukovDoubleDebiasedMachine2018b,neweyCrossFittingFastRemainder2018}. 
In several important applications, however, a form of nonregularity keeps the dependence from vanishing, and conventional confidence intervals undercover: from the classic cross-validation problem of comparing the predictive performance of fitted models \citep{stoneAsymptoticsCrossvalidation1977,efronImprovementsCrossValidation632+1997,bayleRelativeInstabilityModel2026}, to testing for treatment effect heterogeneity \citep{chernozhukovFisherSchultzLecture2025a,wagerCommentFisherSchultzLecture2025a}, to learning the value of a potentially non-unique optimal treatment regime \citep{robinsOptimalStructuralNested2004a,luedtkeStatisticalInferenceMean2016}. 

I revisit the asymptotics of cross-fitting and make two contributions. 
First, exploiting a new locality argument, I show that a large class of split-sample estimators still satisfies a central limit theorem despite the nonregularity, but with an asymptotic variance that depends on the cross-fold correlation.  
Second, I propose a method for estimating this correlation and constructing confidence intervals based on cross-fitting that attain asymptotically nominal coverage in many nonregular settings.

To introduce the problem I study, and previous results, consider a simple example: estimating the conditional predictive performance with 2-fold cross-validation \citep{efronEstimatingErrorRate1983,picardCrossValidationRegressionModels1984,efronImprovementsCrossValidation632+1997,blumBeatingHoldoutBounds1999,dudoitAsymptoticsCrossvalidatedRisk2005a,yangComparingLearningMethods2006,goldsteinTestingRelativePerformance2016,bayleCrossvalidationConfidenceIntervals2020,leiCrossValidationConfidence2020,austernAsymptoticsCrossvalidation2025,luoLimitsAssumptionfreeTests2025}, a routine exercise in applications such as clinical prediction \citep{jiangCalculatingConfidenceIntervals2008,pirracchioMortalityPredictionIntensive2015}. 
Given an i.i.d. sample of an outcome $Y$ and covariates $X$, $(Y_i, X_i) \sim P$ for $i = 1, \dots, n$, an analyst randomly splits $\bc{1,\dots,n}$ into two folds $\s_1$ and $\s_2$ of equal size, and uses a learning algorithm, such as a random forest, to train a model $\muhat_k$ to predict $Y$ from $X$ on each fold $\s_k$, $k = 1, 2$.
The analyst then evaluates each predictor on the other fold, and averages the two estimates: 
$$\wh{\mathrm{Err}} = \frac{1}{2} \bp{ \wh{\mathrm{Err}}_1 + \wh{\mathrm{Err}}_2 }, \qquad \wh{\mathrm{Err}}_k = \frac{1}{n/2} \sum_{i \notin \s_k} L\of{Y_i, \muhat_k(X_i)},$$
where $L$ is a loss function, such as $L(y, \hat{y}) = (y - \hat{y})^2$.
A natural and popular target is the out-of-sample loss of the fitted predictors \citep{efronEstimatingErrorRate1983,dudoitAsymptoticsCrossvalidatedRisk2005a},
$$\mathrm{Err} = \frac{1}{2} \bp{ \mathrm{Err}_1 + \mathrm{Err}_2 }, \qquad \mathrm{Err}_k = \E{ L\of{\tilde{Y}, \muhat_k(\tilde{X})} \;\middle|\; \bc{Y_i, X_i}_{i \in \s_k} },$$
where $(\tilde{Y}, \tilde{X}) \sim P$ independent of the sample. 

Conventional inference on $\mathrm{Err}$ uses a concentration argument that renders cross-fold dependence negligible. 
If $L\of{Y_i, \muhat_k(X_i)} = \bar{L}_n\of{Y_i, X_i} + o_p(1)$, with $\bar{L}_n\of{y, x} = \E{L\of{y, \muhat_k(x)}}$, then, under standard regularity conditions, 
\begin{equation} \label{eq.cv_concentration}
    \sqrt{n/2} \bp{\wh{\mathrm{Err}}_k - \mathrm{Err}_k} = \frac{1}{\sqrt{n/2}} \sum_{i \notin \s_k} \bc{\bar{L}_n\of{Y_i, X_i} - \E{\bar{L}_n\of{\tilde{Y}, \tilde{X}}}} + o_p(1),
\end{equation}
cross-fold dependence is first-order negligible, $\wh{\mathrm{Err}}$ is approximately an average of i.i.d. terms, and $\sqrt{n} (\wh{\mathrm{Err}} - \mathrm{Err})$ is asymptotically normal. 
This is the core mechanism in much of the cross-validation literature with random centering \citep{bayleCrossvalidationConfidenceIntervals2020,austernAsymptoticsCrossvalidation2025,leiModernTheoryCrossValidation2025}, which often verifies concentration around $\bar{L}_n$ through algorithmic stability, and is very close to the arguments in \citet{hubbardStatisticalInferenceData2016b} and \citet{favaTrainingTestingMultiple2025}. 

In practice, however, $\mathrm{Err}$ is usually not of interest on its own, but in comparison with the error of another model or a benchmark. 
Let $\nuhat_k$ denote another model trained on fold $\s_k$ using a different learner, such as a neural network or the sample average $\nuhat_k(x) = (n/2)^{-1} \sum_{i \in \s_k} Y_i$. 
In this case, the cross-validation estimator is 
$$\wh{\Delta} = \frac{1}{2} \bp{ \wh{\Delta}_1 + \wh{\Delta}_2 }, \qquad \wh{\Delta}_k = \frac{1}{n/2} \sum_{i \notin \s_k} \wh{\delta}_k\of{Y_i, X_i},$$
$$\wh{\delta}_k(y, x) = L\of{y, \muhat_k(x)} - L\of{y, \nuhat_k(x)},$$
and the estimand is 
$$\Delta = \frac{1}{2} \bp{\Delta_1 + \Delta_2}, \qquad \Delta_k = \E{ \wh{\delta}_k\of{\tilde{Y}, \tilde{X}} \;\middle|\; \bc{Y_i, X_i}_{i \in \s_k} }.$$

In this setting, the concentration argument in \cref{eq.cv_concentration} does not immediately apply, especially under the leading null hypothesis $\Delta = 0$. 
If both $\muhat_k$ and $\nuhat_k$ converge to the same limit, $\wh{\delta}_k(y, x) = o_p(1)$, and hence 
$\sqrt{n} ({\wh{\Delta} - \Delta}) = o_p(1)$, which is uninformative for inference. 
A more useful result would require a deterministic function $\bar{\delta}_n(y,x) = \mathbb{E}[\wh{\delta}_k(y,x)]$ such that 
\begin{equation} \label{eq.cv_degenerate}
    \frac{\wh{\delta}_k(Y_i, X_i)}{\sigma_n} = \frac{\bar{\delta}_n(Y_i, X_i)}{\sigma_n} + o_p(1),
\end{equation}
with $\sigma_n^2 = \operatorname{Var}[\wh{\delta}_k(\tilde{Y}, \tilde{X})]$, $\sigma_n^2 \to 0$. 
This requirement is counterintuitive: if for example $E[Y|X]$ is constant, the difference in performances $\wh{\delta}_k(Y_i, X_i)$ is pure sampling noise, and it would not be unnatural that $\bar{\delta}_n(y,x) = 0$. 

The failure of \cref{eq.cv_degenerate} is a known form of nonregularity \citep{robinsOptimalStructuralNested2004a,laberAdaptiveConfidenceIntervals2011a}, previously addressed with procedures other than cross-fitting. 
\citet{yadlowskyEvaluatingTreatmentPrioritization2025b} use a single train/holdout split. 
Sequential schemes evaluate each observation with a rule trained only on earlier data \citep{luedtkeStatisticalInferenceMean2016,wagerCommentFisherSchultzLecture2025a}, and hence train and evaluate models on less data than cross-fitting does. 
Repeated sample-splitting combined with rules such as the median or twice the median of $p$-values, or majority votes over intervals, controls size \citep{vandewielTestingPredictionError2009,diciccioExactTestsMultiple2020,gasparinMergingUncertaintySets2024,gasparinCombiningExchangeableValues2025,chernozhukovFisherSchultzLecture2025a}, as do confidence intervals that inflate the standard error or add noise \citep{rinaldoBootstrappingSampleSplitting2019,daiSignificanceTestsFeature2024,verdinelliDecorrelatedVariableImportance2024,bayleRelativeInstabilityModel2026} and the multivariate approach of \citet{favaTrainingTestingMultiple2025}, but all of these procedures are conservative and/or sacrifice power.

I focus instead on cross-fitting itself, and construct confidence intervals that attain asymptotically nominal coverage by exploiting a novel locality condition. 
My first contribution is a central limit theorem (CLT) that accommodates this nonregularity. 
My CLT applies to a broad class of split-sample estimators, covering $K$-fold and leave-one-out cross-fitting. 
Unlike in previous cross-validation CLTs, the asymptotic variance in my result depends on the correlation between fold-level estimators, which turns out to be positive in simulations, so that omitting it would lead to undercoverage. 

My CLT builds on the intuition that the average of dependent random variables can still be asymptotically normal as long as the dependence structure is ``local''. 
In the model comparison example, 
$$\frac{\wh{\Delta}}{\sigma_n} = \frac{1}{n} \bp{\sum_{i \in \s_1} \frac{\wh{\delta}_2\of{Y_i, X_i}}{\sigma_n} + \sum_{j \in \s_2} \frac{\wh{\delta}_1\of{Y_j, X_j}}{\sigma_n}},$$
and the threat to normality is the statistical dependence between $\wh{\delta}_2(Y_i, X_i)$ and $\wh{\delta}_1(Y_j, X_j)$ for $i \in \s_1$ and $ j \in \s_2$. 
Dependence, however, need not invalidate normality. 
Indeed, there is a long tradition of CLTs for dependent data. 
In the time series literature, for example, it is typical to control dependence by assuming that two observations are ``nearly independent'' if they are ``far apart'' in time. 
I formalize the notion of locality in the cross-fitting setting: changing an observation $(Y_i, X_i)$ with $i \in \s_1$ can affect all of $\{\wh{\delta}_1(Y_j, X_j)\}_{j \in \s_2}$, but this effect must be ``small'' if $(Y_i, X_i)$ is ``far'' from $(Y_j, X_j)$. 
This is natural for nonparametric learners: predictions $\muhat_1(X_j)$ are more sensitive to a training point $(Y_i, X_i)$ when $X_j$ is close to $X_i$. 
But the locality condition is not universal: for parametric learners, changing one training point can affect all test points equally.

My CLT paves the way for inference, but raises a new challenge: estimating the cross-fold correlation. 
My second contribution is to address this problem. 
I propose a resampling-based estimator that applies to any black-box learner and many estimators, achieving nominal coverage asymptotically under mild regularity conditions.

\Cref{sec.setup} lays out the setting and notation, and \Cref{sec.clt_2cf} derives the central limit theorem. 
\Cref{sec.inference} is devoted to inference and feasible confidence intervals, and \Cref{sec.simulations} contains a Monte Carlo study supporting the theoretical results. 
\Cref{sec.conclusion} concludes. 

The main text focuses on two-fold cross-fitting to simplify the exposition. 
In \Cref{app.clt_general}, I extend the result to $K$-fold cross-fitting, including leave-one-out. 
All proofs are deferred to the appendix.

\section{Setting and Notation}
\label{sec.setup}

I define the setting and notation, introduce the cross-fitting estimator and its estimand, and present three applications that fit into this framework, including the cross-validation model comparison problem discussed in the introduction. 

Consider an i.i.d. sample $\Dc = \{W_i\}_{i=1}^n$ of even size $n$, with $W_i \sim P$ taking values in $\Wc$. 
For example, $W = (Y, X)$ for an outcome $Y$ and covariates $X$. 
Further, consider a random split of the indices $\bs{n} = \bc{1,\dots,n}$ into two disjoint folds $\s_1,\s_2$ of equal size, independently of the data.
Without loss of generality, I treat the split as fixed.
For any subset $\s \subseteq \bs{n}$, let $\Dc_\s = \bc{W_i : i \in \s}$ denote the corresponding subsample and $\stilde = \bs{n} \setminus \s$ the complement of $\s$.

Let $\Ac$ denote a training algorithm, a measurable map $\Ac : \bigcup_{m \ge 1} \Wc^m \to \Hc$, invariant under permutations of its arguments, that takes a training sample $\Dc$ of any size and outputs a model $\Ac(\Dc) \in \Hc$. 
Moreover, let $\bc{f_\eta : \eta \in \Hc}$ be a known family of functions $f_\eta : \Wc \to \Rb$, and let $\etahat_k = \Ac(\Dc_{\s_k})$ be the model trained on the $k$-th fold. 

I study the cross-fitting estimator 
$$\thetahat_\etahat = \frac{1}{2} \sum_{k=1}^{2} \thetahat_{\etahat_{k}}, \qquad \thetahat_{\etahat_{k}} = \frac{1}{n/2} \sum_{i \notin \s_{k}} f_{{\etahat}_{k}}(W_i),$$
where $\etahat = (\etahat_1, \etahat_2)$. 
The estimand is 
$$\theta_\etahat = \frac{1}{2} \sum_{k=1}^{2} \theta_{\etahat_k}, \qquad \theta_{\etahat_k} = P f_{\etahat_{k}} = \int f_{\etahat_{k}}(w) \rmd P(w).$$
For example, if $f_\eta(w) = (y - \eta(x))^2$, $\theta_\etahat$ is the out-of-sample mean squared error of a predictor that selects one of the two fitted models uniformly at random.

The three examples below fit into this setting and exhibit the nonregularity that motivates this paper.
\begin{example}[Model comparison with cross-validation \citep{blumBeatingHoldoutBounds1999,dudoitAsymptoticsCrossvalidatedRisk2005a,bayleCrossvalidationConfidenceIntervals2020,austernAsymptoticsCrossvalidation2025,favaTrainingTestingMultiple2025,bayleRelativeInstabilityModel2026}] \label{ex.setup.cv}
\,

Let $W = (Y,X)$, with outcome $Y$ and covariates $X$, and let $\eta = (\mu, \nu)$ be a pair of predictors of $Y$ from $X$. 
The fitted model $\etahat_k = (\muhat_k, \nuhat_k)$ then consists of two predictors trained on fold $\s_k$ by different learners, for example a random forest and a neural network. 
The score is the loss difference 
$$f_{\eta}(y,x) = L(y, \mu(x)) - L(y, \nu(x)),$$
where $L$ is a loss function, such as the squared error $L(y,\hat{y}) = (y - \hat{y})^2$ or the classification error $L(y,\hat{y}) = \ind{y \neq \hat{y}}$ for discrete outcomes. 

Each 
$$\theta_{\etahat_k} = \int \bs{L(y, \muhat_k(x)) - L(y, \nuhat_k(x))} \rmd P(y,x)$$
is the difference in out-of-sample loss between the two fitted models $\muhat_k$ and $\nuhat_k$, and the estimand 
$$\theta_\etahat = \frac{1}{2} \bp{\theta_{\etahat_1} + \theta_{\etahat_2}}$$
is its average over the two folds, or, equivalently, the out-of-sample loss difference when $(\muhat_k, \nuhat_k)$ is drawn uniformly at random from the two fitted pairs. 
\end{example}

\begin{example}[Value of a learned treatment rule \citep{luedtkeStatisticalInferenceMean2016,shiBreakingCurseNonregularity2020,whitehouseInferenceOptimalPolicy2026}]
\label{ex.setup.otr}
\,

Let $W = (Y, A, X)$, with outcome $Y$, binary treatment assignment $A$, and covariates $X$, in a randomized experiment with known treatment probability $\P{A = 1 \mid X} = \pi \in (0,1)$.
Moreover, let $Y(0), Y(1)$ denote potential outcomes with $Y = Y(A)$, and let $\eta : \Xc \to \bc{0,1}$ be a treatment rule, so that $\etahat_k$ is a rule learned on fold $\s_k$, for example by empirical welfare maximization \citep{kitagawaWhoShouldBe2018} or by thresholding an estimate of the conditional average treatment effect $\tau(x) = \E{Y(1) - Y(0) \mid X = x}$.

The value of a rule is $V(\eta) = \E{Y(\eta(X))}$, and $V(\eta) = P f_\eta$ for
$$f_\eta(y, a, x) = \frac{a y}{\pi} \eta(x) + \frac{(1-a) y}{1 - \pi} \bc{1 - \eta(x)}.$$
Each $\theta_{\etahat_k} = V(\etahat_k)$ is the value of the rule learned on fold $\s_k$, and the estimand $\theta_\etahat$ is the average value of the two learned rules.

The value of the learned treatment regimes is also informative for the value of the optimal treatment regime. 
The learned rules are dominated by the optimal rule $\eta^*(x) = \ind{\tau(x) > 0}$, so $\theta_\etahat \le V(\eta^*)$, and a lower confidence bound for $\theta_\etahat$ is also a lower confidence bound for the optimal value.
Moreover, inference on $\theta_\etahat$ is asymptotically equivalent to inference on $V(\eta^*)$ whenever the regret $V(\eta^*) - \theta_\etahat = o_p(n^{-1/2})$, which \citet{luedtkeStatisticalInferenceMean2016} establish under margin and rate conditions, without requiring $\etahat_k$ to converge. 
\end{example}

\begin{example}[Best linear predictor (BLP) test of predictive power \citep{chernozhukovFisherSchultzLecture2025a,wagerCommentFisherSchultzLecture2025a}]
\label{ex.setup.blp}
\,

Let $W = (Y, X)$, and let $\eta : \Xc \to \Rb$ be a predictor of $Y$ from $X$.
A natural way to assess the predictive power and calibration of $\etahat_k$ is to regress $Y$ on $\etahat_k(X)$ and a constant in the test fold $\stilde_k$. 
A slope of zero means no predictive power, and a slope of one means good calibration. 

Holding $\etahat_k$ fixed, the population slope is 
$$\theta_{\etahat_k} = Pf_{\etahat_k}, \qquad f_\eta(y, x) = \frac{(y - \E{Y})\bc{\eta(x) - P\eta}}{P(\eta - P\eta)^2}.$$ 
The ordinary least squares estimator differs from $\thetahat_{\etahat_k}$ due to the oracle scale $P(\etahat_k - P\etahat_k)^2$ and the oracle centerings $\E{Y}$ and $P\eta$, but this distinction is typically first-order negligible since the sample variance of $\etahat_k(X)$ is ratio-consistent for $P(\etahat_k - P\etahat_k)^2$ under mild moment conditions. 
Under the null that $\E{Y \mid X}$ is constant, if $\etahat_k$ is consistent for this constant, $P(\etahat_k - P\etahat_k)^2 \pto 0$, and the score diverges. 
The standardized score $(f_{\etahat_k} - P f_{\etahat_k}) / \operatorname{Var}[{f_{\etahat_k}(W)}]^{1/2}$ may then have no deterministic limit, which is the nonregularity my CLT addresses.

The same construction tests for treatment effect heterogeneity in a randomized controlled trial by replacing $Y$ with its Horvitz--Thompson transform, which yields the best linear predictor of \citet[Theorem~3.2]{chernozhukovFisherSchultzLecture2025a}. 
\end{example}

\section{A CLT for Two-Fold Cross-Fitting}
\label{sec.clt_2cf}

I establish a central limit theorem for two-fold cross-fitting that accommodates the examples of the previous section. 
First, I discuss the required assumptions of the CLT, including a novel locality condition. 
Then, I establish a qualitative CLT, followed by a quantitative result that establishes convergence to normality at the i.i.d. rate $n^{-1/2}$ in the Wasserstein metric under stronger conditions. 

To simplify exposition, I focus on data-generating processes (DGPs) where 
$$P f_{\eta}=0 \qquad \text{and} \qquad P f_{\eta}^2 = 1$$ 
for every $\eta \in \Hc$. 
Alternatively, the general result follows from replacing each score $f_{\etahat_k}(w)$ by $(f_{\etahat_k}(w) - P f_{\etahat_k}) / \sigma_{\etahat_k}$, with $\sigma^2_{\eta} = P(f_\eta - P f_\eta)^2$. 
While this normalization is infeasible, it is a natural target for a CLT, since the centering reflects the estimand of interest $\theta_\etahat$, and $\sigma^2_{\etahat_k}$ can often be estimated consistently from the test sample $\stilde_k$, which I address when deriving feasible standard errors. 
Importantly, this normalization allows the variance of the unstandardized score $\sigma^2_{\etahat_k}$ to converge to zero in probability, requiring only that it be almost surely positive for every $n$. 

This section's main result is a new CLT establishing 
$$T_n = \frac{\sqrt{n} \thetahat_\etahat}{\sqrt{1 + \rho_n}} \leadsto \Nc \of{0, 1},$$
with 
$$\rho_n = \Corr{\thetahat_{\etahat_1}, \thetahat_{\etahat_2}} > -1.$$
I derive feasible standard errors in \Cref{sec.inference}.

I collect in \Cref{ass.clt_2cf} the conditions I use for the CLT. 
I use the convention, without loss of generality, that $1,2 \in \s_1$. 
The assumption uses three further definitions. 
First, define the mean-square stability coefficient 
$$\gamma_n = \E{\bc{f_{\etahat_1}(W) - f_{\etahat'_1}(W)}^2 },$$
where $W \sim P$ independent of the data, $\Dc'_{\s_1} = \bc{W'_1} \cup \bc{W_i : i \in \s_1 \setminus \bc{1}}$ is the training sample with $W_1$ replaced by an independent draw $W'_1 \sim P$, and $\etahat'_1 = \Ac(\Dc'_{\s_1})$ is the corresponding refitted model.
$\gamma_n$ is the mean-square stability coefficient studied by \citet{kaleCrossValidationMeanSquareStability2011} and \citet{bayleCrossvalidationConfidenceIntervals2020}, applied here to the centered and conditionally normalized score.
Second, define the local covariance coefficients 
$$r_i = \frac{n}{2} \Cov{\sqrt{\frac{n}{2}} \thetahat_{\etahat_1}, \sqrt{\frac{n}{2}} \thetahat_{\etahat_2} \;\middle|\; \Dc_{-i} }, \qquad i=1,\dots,n,$$
where $\Dc_{-i} = \bc{W_j : j \neq i}$ contains all observations except $W_i$. 
Each coefficient measures the covariance between the fold estimators when only observation $i$ varies, isolating the interaction between using $W_i$ to train $\etahat_1$ and to evaluate $\thetahat_{\etahat_2}$, when $i \in \s_1$. 
Note that this normalization makes $r_i$ a local analogue of the cross-fold correlation, since 
$$\E{r_i} = \rho_n.$$
Moreover, for a random variable $X$, write
$$
\norm{X}_p
=
\begin{cases}
\E{\abs{X}^p}^{1/p}, & 1 \le p < \infty,\\
\inf\bc{M \ge 0 : \P{\abs{X} \le M} = 1}, & p = \infty.
\end{cases}
$$

\begin{assumption}%
\label{ass.clt_2cf}
\,
For some $\bar{c} < \infty$, $\ubar{c} > 0$, $p \in (2, \infty]$, with the convention $1/p = 0$ when $p = \infty$, and sufficiently large $n$, the following conditions hold:
\begin{assumptionenum}
    \item \label{ass.clt_2cf.moments} (Moments). 
    $$\norm{f_{\etahat_1}(W)}_p \le \bar{c};$$
    \item \label{ass.clt_2cf.degeneracy} (Non-degeneracy). 
    $$\rho_n \ge -1 + \ubar{c};$$
    \item \label{ass.clt_2cf.stability} (Mean-square stability). 
    $$n^{\frac{1}{2(1-1/p)}} \gamma_n \to 0;$$
    \item \label{ass.clt_2cf.locality} (Locality).
    $$\Cov{r_1, r_2} \to 0.$$
\end{assumptionenum}
\end{assumption}

\begin{lemma}[Sufficient Conditions]
\label{lem.clt_2cf.sufficient}
    When $\Var{r_1} > 0$, 
    $$\abs{\Cov{r_1, r_2}} \le n \gamma_n \abs{\Corr{r_1, r_2}}.$$
    In particular, \Cref{ass.clt_2cf.stability,ass.clt_2cf.locality} hold under any of the following:
    \begin{enumerate}[label=\upshape(\alph*)]
        \item %
        $n \gamma_n = O(1), \Corr{r_1, r_2} = o(1)$;
        \item %
        $p = \infty, \sqrt{n} \, \gamma_n = o(1), \sqrt{n} \, \Corr{r_1, r_2} = O(1)$;
        \item %
        $n \gamma_n = o(1)$.
    \end{enumerate}
\end{lemma}

\Cref{ass.clt_2cf.moments} is a uniform $L^p$ moment bound for some $p > 2$, with $p = \infty$ for bounded scores. 
\Cref{ass.clt_2cf.degeneracy} rules out near-perfect negative correlation between the fold-level estimators. 
It holds, for example, whenever $\rho_n \ge 0$, as in all simulations of \Cref{sec.simulations}.  
\Cref{ass.clt_2cf.stability} is a mean-square stability condition, restricting the sensitivity of the fitted score to replacing one training observation. 
Its exponent $1/[2(1-1/p)]$ decreases from $1$ as $p \downarrow 2$ to $\frac{1}{2}$ at $p = \infty$, so that lighter tails permit a less stable learner. 
Related conditions have been used to establish asymptotic normality of cross-validation estimators of test error under the stronger requirement $n \gamma_n = o(1)$ \citep{bayleCrossvalidationConfidenceIntervals2020,austernAsymptoticsCrossvalidation2025,leiModernTheoryCrossValidation2025}, which implies $\rho_n \to 0$ (\Cref{remark.stability}). 
This condition can often be verified using algorithm-specific leave-one-out or replace-one stability bounds, available under suitable regularity and tuning conditions \citep[for a discussion, see, e.g.,][Section~6]{leiModernTheoryCrossValidation2025}. 

Finally, \Cref{ass.clt_2cf.locality} is a locality condition. 
It requires the dependence created by using the same data point $W_i$ to train model $\etahat_1$ and evaluate $f_{\etahat_{2}}(W_i)$ to be local across observations. 
Note that $r_i$ measures the dependence between $\thetahat_{\etahat_1}$ and $\thetahat_{\etahat_2}$ as only observation $W_i$ varies, isolating the interaction between using $W_i$ to train $\etahat_1$ and to evaluate $\thetahat_{\etahat_2}$, when $i \in \s_1$. 
Under the stability requirement $n \gamma_n = O(1)$, it suffices that $\Corr{r_1, r_2} \to 0$ at any rate (\Cref{lem.clt_2cf.sufficient}(a)). 
This condition is plausible for nonparametric learners, including machine learning algorithms, because changing a training observation mostly affects predictions near its own covariate value.
The coefficient $r_i$ is then driven by the test observations whose covariates lie close to those of $W_i$, so that $r_1$ and $r_2$ depend on nearly disjoint parts of the test fold and are approximately uncorrelated.
In parametric models, however, all observations act through the same few parameters, so the contributions $r_i$ can remain correlated, and asymptotic normality may fail. 
Note that $\E{r_i} = \rho_n$, and that \Cref{ass.clt_2cf.locality} allows $\rho_n$ to be non-zero asymptotically.

\Cref{lem.clt_2cf.sufficient} gives three sufficient conditions for \Cref{ass.clt_2cf.stability,ass.clt_2cf.locality}, trading off stability against locality. 
Since $\Var{r_1} \le n \gamma_n$, \Cref{ass.clt_2cf.locality} holds whenever $n \gamma_n \Corr{r_1, r_2} \to 0$. 
Under condition (a), stability alone keeps $n \gamma_n$ bounded, so the correlation may vanish at any rate. 
Condition (b) relaxes stability by strengthening locality: once the correlation vanishes at rate $n^{-1/2}$, stability is needed only at the $\sqrt{n}$ rate for bounded scores, $\sqrt{n} \gamma_n \to 0$. 
Condition (c) is the stability requirement of previous cross-validation CLTs, and it implies \Cref{ass.clt_2cf.locality} with no restriction on the correlation. 
Hence, my CLT nests the regime $n \gamma_n \to 0$, in which $\rho_n \to 0$. %

\begin{theorem}[Cross-Fitting Central Limit Theorem]
\label{thm.clt_2cf} 
Let \Cref{ass.clt_2cf} hold. 
Then, 
$$T_n \leadsto \Nc \of{0, 1}.$$
\end{theorem}

\Cref{thm.clt_2cf} is, to my knowledge, the first to establish a central limit theorem for cross-fitting estimators that accommodates the nonregularity that the standardized influence functions may not converge to a fixed limit. 
Compared to previous CLTs in the cross-validation \citep{dudoitAsymptoticsCrossvalidatedRisk2005a,bayleCrossvalidationConfidenceIntervals2020,austernAsymptoticsCrossvalidation2025,leiModernTheoryCrossValidation2025} and semiparametric estimation/double machine learning \citep{schickAsymptoticallyEfficientEstimation1986,chernozhukovDoubleDebiasedMachine2018b} literatures, the asymptotic variance of $\sqrt{n} \thetahat_\etahat$ in my result depends on the cross-fold correlation, which may be nonzero. 
Note that the stronger stability condition $n \gamma_n = o(1)$, required by previous cross-validation CLTs, implies $\rho_n = o(1)$. 
Compared to those results, I weaken the stability requirement (\Cref{ass.clt_2cf.stability}), and introduce the locality \Cref{ass.clt_2cf.locality}. 

The next result strengthens \Cref{thm.clt_2cf}. 
Under stronger conditions, the normal approximation holds at the rate $n^{-1/2}$, the Berry--Esseen rate for averages of independent random variables with finite third moments. 
I state the result in the Wasserstein distance, defined for random variables $X$ and $Y$ by 
$$d_W\of{X, Y} = \sup_{h \in \Lip{1}} \, \bigl| \E{h(X)} - \E{h(Y)} \bigr|,$$
where 
$$\Lip{1} = \bc{h \colon \Rb \to \Rb \;:\; \abs{h(x) - h(y)} \le \abs{x - y} \text{ for all } x, y \in \Rb}.$$ 

\begin{theorem}[Quantitative Cross-Fitting Central Limit Theorem]
\label{thm.clt_2cf.quantitative}
Let \Cref{ass.clt_2cf} hold with $p \ge 3$.
Suppose, in addition, that
$$
n^{\frac{1}{1-1/p}}\gamma_n=O(1),
\qquad 
n \Cov{r_1, r_2}=O(1).
$$
Then,
$$
d_W\of{T_n, \Nc\of{0,1}}=O(n^{-1/2}).
$$
\end{theorem}

\Cref{thm.clt_2cf.quantitative} shows that the dependence between folds does not slow the normal approximation as long as the locality condition holds at a fast enough rate. 
Its conditions trade off tails against stability. 
The moment condition requires a uniformly bounded $L^p$ norm of the score with $p \ge 3$, as in the Berry--Esseen theorem. 
The stability condition requires $\gamma_n = O(n^{-1/(1-1/p)})$, which ranges from $\gamma_n = O(n^{-3/2})$ at $p = 3$ to $\gamma_n = O(n^{-1})$ for bounded scores, so that heavier tails require a more stable learner. 
For bounded scores, the stability condition has the same order as the sufficient condition $n \gamma_n = O(1)$ in \Cref{lem.clt_2cf.sufficient}(a), and the $O(n^{-1/2})$ rate then costs only a strengthening of locality, from $\Cov{r_1, r_2} = o(1)$ to $O(n^{-1})$. 
Since $\Cov{r_1, r_2} = \Var{r_1} \Corr{r_1, r_2}$ and $\Var{r_1} \le n \gamma_n$ (\Cref{lem.clt_2cf.stability_bounds}), the latter holds whenever $\Corr{r_1, r_2} = O(n^{-1})$.

\begin{remark}[Comparison with previous stability conditions]
\label{remark.stability}
Previous stability-based CLTs for cross-validation require first-order loss stability at the $o(n^{-1})$ scale relative to the variance of the leading score
\citep{bayleCrossvalidationConfidenceIntervals2020,austernAsymptoticsCrossvalidation2025,leiModernTheoryCrossValidation2025}.
Applied to the centered and conditionally normalized score used here, this condition is equivalent to
\[
n\gamma_n\to0.
\]
Indeed, if $\widetilde\sigma_n^2=\operatorname{Var}[\mathbb{E}[f_{\etahat_1}(W)\mid W]]$, their condition is
$n\gamma_n/\widetilde\sigma_n^2\to0$, while
\[
0\le 1-\widetilde\sigma_n^2\le \frac n4\gamma_n.
\]
Hence, previous results cover the regime $n\gamma_n=o(1)$, which also forces $\rho_n\to0$.
By contrast, \Cref{thm.clt_2cf} allows weaker stability, for example $n\gamma_n=O(1)$ when the locality condition makes the covariance contributions asymptotically uncorrelated.
\end{remark}

\section{Estimation of \texorpdfstring{$\rho$}{rho} and Inference}
\label{sec.inference}

I propose a resampling-based procedure to construct feasible confidence intervals based on \Cref{thm.clt_2cf}, and show that they attain nominal coverage asymptotically under regularity conditions. 

Let $n=4m$ be the total sample size, and return to the unnormalized scores of \Cref{sec.setup}, that is, $f_{\etahat_k}(w)$ may have non-zero mean and non-unit variance. 
Throughout this section, $\rho_n$ denotes the cross-fold correlation of the centered and normalized fold estimators,
$$\rho_n = \Corr{\frac{\thetahat_{\etahat_1} - \theta_{\etahat_1}}{\sigma_{\etahat_1}}, \frac{\thetahat_{\etahat_2} - \theta_{\etahat_2}}{\sigma_{\etahat_2}}},$$
which is the correlation of \Cref{sec.clt_2cf} computed with each score $f_{\etahat_k}$ replaced by $(f_{\etahat_k} - \theta_{\etahat_k}) / \sigma_{\etahat_k}$. %
I consider directly observable scores, which fit \Cref{ex.setup.cv,ex.setup.otr}, and give an implementation for the BLP \Cref{ex.setup.blp} in \Cref{app.inference.blp}.
Keep the original folds fixed.
For $b=1,\dots,B$ and $k=1,2$, independently draw a uniform subset $\s_{k,b}\subset\s_k$ of size $m$ and fit
$$
\etahat_{k,b}=\Ac(\Dc_{\s_{k,b}}).
$$
Define
\begin{align}
\sigmahat_{\etahat_{k,b}}^2
&=
\frac1{2m-1}\sum_{i\in\stilde_k}
\bs{
f_{\etahat_{k,b}}(W_i)
-
\frac1{2m}\sum_{j\in\stilde_k}f_{\etahat_{k,b}}(W_j)
}^2,
\label{eq.inference.pooled_scale}\\
Z_{k,b}
&=
\frac{
\sum_{i\in\s_{3-k,b}}f_{\etahat_{k,b}}(W_i)
-
\sum_{i\in\stilde_k\setminus\s_{3-k,b}}f_{\etahat_{k,b}}(W_i)
}{
\sqrt m\,\sigmahat_{\etahat_{k,b}}
},
\label{eq.inference.contrast}\\
\wh\rho_n
&=
\max\bc{
-1,\,
\min\bc{
1,\,
\frac1B\sum_{b=1}^B Z_{1,b}Z_{2,b}
}
}.
\label{eq.inference.rhohat}
\end{align}
Set $Z_{k,b}=0$ if its denominator is zero.
Note that $Z_{k,b}$ evaluates the same fitted score on two disjoint samples, so its expected value is zero. 

Let $\sigmahat_{\etahat_k}^2$ be the test-fold sample variance
\begin{equation}
\label{eq.inference.original_scale}
\sigmahat_{\etahat_k}^2
= \frac{1}{2m-1}\sum_{i\in\stilde_k}
\bs{f_{\etahat_k}(W_i)-\thetahat_{\etahat_k}}^2,
\qquad k=1,2.
\end{equation}
The proposed standard error and confidence interval are
\begin{align}
\wh{\se}\of{\thetahat_\etahat}
&=
\sqrt{
\frac{
\sigmahat_{\etahat_1}^2+\sigmahat_{\etahat_2}^2
+2\wh\rho_n\sigmahat_{\etahat_1}\sigmahat_{\etahat_2}
}{2n}
},
\label{eq.inference.se}\\
\CIalpha
&=
\bs{
\thetahat_\etahat
\pm z_{1-\alpha/2}\wh{\se}\of{\thetahat_\etahat}
},
\label{eq.inference.ci}
\end{align}
where $z_{1-\alpha/2}$ is the standard normal quantile.
The original estimator and its target are unchanged.

\begin{theorem}[Asymptotic Validity of Confidence Interval]
\label{thm.inference}
Let \Cref{ass.clt_2cf} hold after centering and normalizing the scores. 
Suppose $0 < \sigma_\eta^2 < \infty$ for every $\eta\in\Hc$.
If, for some $\rho > - 1$, 
\begin{equation}
\Var{\E{Z_{1,1}Z_{2,1}\mid\Dc}}\to0,
\qquad
\rho_{n} \to \rho,
\qquad
B\to\infty,
\label{eq.inference.resampling}
\end{equation}
then $\wh\rho_n-\rho_n\Lto[2]0$.
If, additionally,
\begin{equation}
\frac{\sigma_{\etahat_1}}{\sigma_{\etahat_2}}\pto1,
\label{eq.inference.scale_balance}
\end{equation}
then, for every fixed $\alpha\in(0,1)$,
$$
\P{\theta_\etahat\in\CIalpha}\to1-\alpha.
$$
\end{theorem}

The first condition in \cref{eq.inference.resampling} requires the covariance averaged over resampling splits to stabilize across datasets.
The second requires the cross-fold correlation to converge to some limit greater than $-1$.

\section{Simulation Study}
\label{sec.simulations}

I examine the finite-sample performance of the confidence intervals derived in \Cref{sec.inference} in a simulation study. 
In particular, I consider the cross-validation (CV) model comparison (\Cref{ex.setup.cv}) and the BLP test of predictive power (\Cref{ex.setup.blp}). 
I use sample sizes $n \in \bc{1000, 2000}$ and 2,400 Monte Carlo replications for each data-generating process and sample size. 
The correlation estimator uses $B=500$ repetitions. 

Observations are independent and identically distributed, with $X_i \sim \Nc(0, I_5)$ and $U_i \sim \Nc(0, 1)$ independent of $X_i$.
I consider two no-signal designs. 
In the homoskedastic Gaussian design,
\[
    Y_i=U_i.
\]
In the heteroskedastic, conditionally skewed design,
\[
    Y_i
    =
    \omega(X_i)
    \frac{
        \exp{\lambda(X_i)U_i}
        -\exp{\lambda(X_i)^2/2}
    }{
        \sqrt{
            \exp{2\lambda(X_i)^2}
            -\exp{\lambda(X_i)^2}
        }
    },
\]
where, writing $\Lambda(t)=\bc{1+\exp{-t}}^{-1}$,
\[
    \omega(x)=0.5+1.5\Lambda(x_1+0.5x_2),
    \qquad
    \lambda(x)=0.25+0.75\Lambda(x_2-x_3).
\]
Both designs satisfy $\E{Y_i \mid X_i}=0$.
The conditional variance is one in the Gaussian design and $\omega(X_i)^2$ in the skewed design. 
The latter also permits conditional skewness to vary with the covariates.

I estimate the predictive model with two learners.
The first is a single-hidden-layer neural network from the \texttt{R} package \texttt{nnet}, with 10 hidden units, a linear output, no skip connections, and weight decay equal to $0.1$.
The second is a random forest from the \texttt{R} package \texttt{ranger}, with 200 trees, three candidate covariates per split, and a minimum node size of five.

For CV, I compare the squared-error loss of the fitted learner with that of the training fold outcome mean.
In the notation of \Cref{ex.setup.cv}, $L(y, \hat{y}) = (y - \hat{y})^2$, $\muhat_k$ is the learner trained on $\s_k$, and $\nuhat_k$ is the constant predictor equal to the average of $Y_i$ over $i \in \s_k$.

For both applications, I report the coverage rate $\P{\theta_\etahat \in \CIalpha}$ at the nominal $95\%$ level.
Since $\E{Y_i \mid X_i} = 0$ in both designs, the BLP estimand is $\theta_\etahat=0$.
For CV, the estimand of \Cref{ex.setup.cv} is
$$\theta_\etahat = \frac{1}{2} \sum_{k=1}^{2} \int \bs{\bc{y - \muhat_k(x)}^2 - \bc{y - \nuhat_k}^2} \rmd P(y, x),$$
the average out-of-sample loss difference, holding the fitted models fixed.
This target need not be zero under a no-signal design.
I compute it by integrating analytically over the outcome noise and numerically over 65,000 independent covariate draws per fitted model. 

\begin{table}[tbp]
    \centering
    \small
    \caption{Coverage rates of nominal 95\% BLP and CV confidence intervals.}
    \label{tab.simulations.coverage}
    \begin{tabular}{@{}lrcccc@{}}
        \toprule
        & & \multicolumn{2}{c}{BLP} & \multicolumn{2}{c}{CV} \\
        \cmidrule(lr){3-4}
        \cmidrule(l){5-6}
        Design & $n$ & Estimated $\wh\rho_n$ & $\wh\rho_n = 0$ & Estimated $\wh\rho_n$ & $\wh\rho_n = 0$ \\
        \midrule
        \multicolumn{6}{c}{Neural network} \\
        \addlinespace
        Gaussian & 1000 & 94.38 & 90.42 & 94.21 & 91.67 \\
                 & 2000 & 95.04 & 91.13 & 94.75 & 91.54 \\
        Skewed   & 1000 & 94.33 & 91.54 & 94.38 & 92.75 \\
                 & 2000 & 95.13 & 92.38 & 94.42 & 92.13 \\
        \midrule
        \multicolumn{6}{c}{Random forest} \\
        \addlinespace
        Gaussian & 1000 & 94.13 & 84.50 & 94.54 & 85.50 \\
                 & 2000 & 95.67 & 86.88 & 95.54 & 87.54 \\
        Skewed   & 1000 & 94.75 & 87.96 & 95.29 & 89.00 \\
                 & 2000 & 95.33 & 89.33 & 95.79 & 90.38 \\
        \bottomrule
    \end{tabular}

    \medskip
    \begin{minipage}{\linewidth}
        \footnotesize
        \textit{Notes:}
        Entries are coverage percentages of $\theta_{\etahat}$,
        based on 2,400 replications per design and sample size.
        Gaussian denotes the homoskedastic Gaussian design,
        Skewed denotes the heteroskedastic, conditionally skewed design.
        Estimated $\wh\rho_n$ uses \cref{eq.inference.rhohat} with $B=500$;
        $\wh\rho_n = 0$ ignores the cross-fold correlation.
    \end{minipage}
\end{table}

\Cref{tab.simulations.coverage} reports the results.
With the estimated correlation, coverage rates range from $94.13\%$ to $95.67\%$ for BLP and from $94.21\%$ to $95.79\%$ for CV, close to the nominal $95\%$ level across the learners, designs, and sample sizes considered.
Setting $\wh\rho_n = 0$, which ignores the cross-fold correlation, yields undercoverage in every case, with coverage rates from $84.50\%$ to $92.38\%$ for BLP and from $85.50\%$ to $92.75\%$ for CV.
The undercoverage is most severe for the random forest, whose coverage falls below $90\%$ in the Gaussian design for both applications and sample sizes.

\section{Conclusion} \label{sec.conclusion}

I revisited the problem of inference with cross-fitting when the standardized fitted score need not converge to a deterministic limit, a nonregularity shared by cross-validation model comparison, tests of predictive power and treatment effect heterogeneity, and the value of learned treatment rules. 
Exploiting a new locality condition, I established a central limit theorem for cross-fitting estimators with an asymptotic variance that depends on the cross-fold correlation, and I proposed confidence intervals based on a resampling estimator of this correlation that attain asymptotically nominal coverage. 
In simulations with random forests and neural networks, these intervals delivered coverage close to the nominal level, whereas intervals that ignore the cross-fold correlation undercovered.

\clearpage
\renewcommand{\baselinestretch}{1}
\bibliography{refs2.bib}
\clearpage

\begin{appendices}

\appendixnumbering
\crefalias{section}{appendix}
\crefalias{subsection}{appendix}
\crefalias{subsubsection}{appendix}

\section{Extension to \texorpdfstring{$K_n$}{K-n}-Fold Cross-Fitting}
\label{app.clt_general}

I extend \Cref{thm.clt_2cf} to a possibly growing number of folds, including leave-one-out cross-fitting.

Let $K_n\in\bc{2,\dots,n}$ divide $n$, let $m_n=n/K_n$, and partition $\bs{n}$ into equal-sized folds
$$
\s_1,\dots,\s_{K_n}.
$$
For $k=1,\dots,K_n$, let
$$
\etahat_{-k}
=
\Ac\bp{\Dc_{\bs{n}\setminus\s_k}}
$$
denote the model trained on all observations outside fold $\s_k$.
In this section, $k(i)$ denotes the fold containing observation $i$, so that $W_i$ is evaluated by $\etahat_{-k(i)}$.
The cross-fitting estimator and its estimand are
$$
\thetahat_\etahat
=
\frac{1}{K_n}\sum_{k=1}^{K_n}\thetahat_{\etahat_{-k}},
\qquad
\thetahat_{\etahat_{-k}}
=
\frac{1}{m_n}\sum_{i\in\s_k}f_{\etahat_{-k}}(W_i),
\qquad
\theta_\etahat
=
\frac{1}{K_n}\sum_{k=1}^{K_n}P f_{\etahat_{-k}},
$$
which reduce to those of \Cref{sec.setup} when $K_n=2$.

Unlike in the two-fold result, I use a deterministic unconditional normalization.
Let
\begin{equation}
\sigma_n^2
=
\E{P\bc{
f_{\etahat_{-1}}-P f_{\etahat_{-1}}
}^2}
\in(0,\infty),
\label{eq.clt_K.sigma}
\end{equation}
and, for every $\eta\in\Hc$, define the standardized score
$$
\tilde f_\eta(w)
=
\frac{f_\eta(w)-P f_\eta}{\sigma_n}.
$$
Thus $P\tilde f_\eta=0$ for every $\eta\in\Hc$, and, by fold symmetry, for $k=1,\dots,K_n$,
\begin{equation}
P \tilde f_{\etahat_{-k}}=0,
\qquad
\E{P \tilde f_{\etahat_{-k}}^2}=1.
\label{eq.clt_K.normalization}
\end{equation}

Define the conditional variances of the standardized scores and their average across folds,
\begin{equation}
v_{k,n}
=
P \tilde f_{\etahat_{-k}}^2,
\qquad
\bar v_n
=
\frac{1}{K_n}\sum_{k=1}^{K_n}v_{k,n}.
\label{eq.clt_K.v}
\end{equation}
Hence $\E{v_{k,n}}=\E{\bar v_n}=1$, but $v_{k,n}$ need not equal one for a realized fit.

For each fold, let
\begin{equation}
S_{k,n}
=
\frac{1}{\sqrt{m_n}}
\sum_{i\in\s_k}
\tilde f_{\etahat_{-k}}(W_i).
\label{eq.clt_K.fold}
\end{equation}
Conditional on $\etahat_{-k}$, the summands are i.i.d.\ with mean zero and variance $v_{k,n}$, so $\E{S_{k,n}}=0$ and $\Var{S_{k,n}}=\E{v_{k,n}}=1$.
By fold symmetry,
$$
\rho_n
=
\Cov{S_{k,n},S_{\ell,n}},
\qquad
k\ne\ell,
$$
does not depend on $k$ and $\ell$.
Since $\sqrt n\bp{\thetahat_\etahat-\theta_\etahat}/\sigma_n=K_n^{-1/2}\sum_{k=1}^{K_n}S_{k,n}$,
\begin{equation}
\Var{\frac{\sqrt n\bp{\thetahat_\etahat-\theta_\etahat}}{\sigma_n}}
=
V_n,
\qquad
V_n=1+(K_n-1)\rho_n.
\label{eq.clt_K.variance}
\end{equation}
Define
\begin{equation}
T_n
=
\frac{
\sqrt n\bp{\thetahat_\etahat-\theta_\etahat}
}{
\sigma_n\sqrt{V_n}
}.
\label{eq.clt_K.T}
\end{equation}

Without loss of generality, label the observations so that 
$$\s_k=\bc{(k-1)m_n+1,\dots,km_n}$$ 
for $k=1,\dots,K_n$.
In particular, $1\in\s_1$, $2\in\s_1$ when $m_n\ge2$, and $m_n+1\in\s_2$.
Let $W_1'\sim P$ be independent of the data, and let $\etahat_{-2}'$ denote the model obtained after replacing $W_1$ by $W_1'$ in the training sample of $\etahat_{-2}$.
Define the mean-square stability coefficient
\begin{equation}
\gamma_n
=
\E{\bc{
\tilde f_{\etahat_{-2}}(W)
-
\tilde f_{\etahat_{-2}'}(W)
}^2},
\label{eq.clt_K.gamma}
\end{equation}
where $W\sim P$ is independent of all observed and replacement data.

For $i=1,\dots,n$, define
\begin{equation}
r_i
=
\sqrt{m_n}\,
\Cov{
\tilde f_{\etahat_{-k(i)}}(W_i),
\frac{1}{K_n-1}
\sum_{\ell\ne k(i)}
S_{\ell,n}
\,\middle|\,
\Dc_{-i}
}.
\label{eq.clt_K.r}
\end{equation}
This is the $K_n$-fold analogue of the local covariance coefficient of \Cref{sec.clt_2cf}, and
\begin{equation}
\E{r_i}=\rho_n.
\label{eq.clt_K.rho_identity}
\end{equation}
Indeed, since $\etahat_{-k(i)}$ is $\Dc_{-i}$-measurable, $\E{\tilde f_{\etahat_{-k(i)}}(W_i)\mid\Dc_{-i}}=0$, so 
$$\E{r_i}=\sqrt{m_n}(K_n-1)^{-1}\sum_{\ell\ne k(i)}\E{\tilde f_{\etahat_{-k(i)}}(W_i)S_{\ell,n}},$$ 
and exchangeability within $\s_{k(i)}$ gives 
$$\sqrt{m_n}\,\E{\tilde f_{\etahat_{-k(i)}}(W_i)S_{\ell,n}}=\Cov{S_{k(i),n},S_{\ell,n}}=\rho_n.$$

The approximate Stein kernel is
\begin{equation}
\Kc_n
=
\frac{\bar v_n+(K_n-1)\bar r_n}{V_n},
\qquad
\bar r_n
=
\frac1n\sum_{i=1}^n r_i,
\label{eq.clt_K.kernel}
\end{equation}
and $\E{\Kc_n}=1$ by \cref{eq.clt_K.v,eq.clt_K.rho_identity}.

\begin{theorem}[$K_n$-Fold Cross-Fitting Central Limit Theorem]
\label{thm.clt_K}
Suppose that there exist $\bar c<\infty$, $\ubar c>0$, and $p\in(2,\infty]$, with $1/p=0$ when $p=\infty$, such that the first two conditions below hold for all sufficiently large $n$ and the last three hold as $n\to\infty$:
\begin{align}
\norm{\tilde f_{\etahat_{-1}}(W)}_p
&\le \bar c,
\label{ass.clt_K.moments}\\
V_n
&\ge \ubar c,
\label{ass.clt_K.degeneracy}\\
\frac{\Var{\bar v_n}}{V_n^2}
&\to0,
\label{ass.clt_K.scale}\\
n^{\frac{1}{2(1-1/p)}}(K_n-1)\gamma_n
&\to0,
\label{ass.clt_K.stability}\\
\frac{(K_n-1)^2\Var{\bar r_n}}{V_n^2}
&\to0.
\label{ass.clt_K.locality}
\end{align}
Then
$$
T_n\leadsto\Nc\of{0,1}.
$$
\end{theorem}

\Cref{ass.clt_K.scale} is an additional condition required by the unconditional normalization.
It requires the average conditional variance $\bar v_n$ to concentrate around its mean.
By fold symmetry, $\Var{\bar v_n}=\Var{v_{1,n}}/K_n+(1-1/K_n)\Cov{v_{1,n},v_{2,n}}$, so, when $K_n\to\infty$, the condition can hold even if each $v_{k,n}$ remains random, provided $\Var{v_{1,n}}/(K_nV_n^2)\to0$ and $\Cov{v_{1,n},v_{2,n}}/V_n^2\to0$.
Under the conditional normalization used in \Cref{sec.clt_2cf}, $v_{k,n}=1$ almost surely, so the condition holds immediately.
\Cref{ass.clt_K.stability} is the analogue of \Cref{ass.clt_2cf.stability}, with the same exponent.
The factor $K_n-1$ reflects that one observation enters the training samples of $K_n-1$ fitted models.
For fixed $K_n$, the condition has the same order as \Cref{ass.clt_2cf.stability}.
\Cref{ass.clt_K.locality} ensures that the local covariance term of \cref{eq.clt_K.kernel} has vanishing variance.
By fold symmetry and exchangeability within folds,
$$
\Var{\bar r_n}
=
\frac{\Var{r_1}}{n}
+
\frac{m_n-1}{n}\Cov{r_1,r_2}
+
\bp{1-\frac{1}{K_n}}\Cov{r_1,r_{m_n+1}},
$$
where the second term vanishes when $m_n=1$.
The three terms control the diagonal contribution, the dependence between distinct observations in the same fold, and the dependence between observations in different folds, respectively.

When $K_n=2$ and the scores satisfy the normalization of \Cref{sec.clt_2cf}, $\sigma_n=1$, $\tilde f_\eta=f_\eta$, and $v_{k,n}=1$ almost surely, so \cref{ass.clt_K.scale} holds trivially.
Moreover, $\rho_n$, $\gamma_n$, and $r_i$, $i\in\s_1$, coincide with those of \Cref{sec.clt_2cf}, with $\etahat_{-1}$ and $\etahat_{-2}$ in the roles of $\etahat_2$ and $\etahat_1$, so \cref{ass.clt_K.moments,ass.clt_K.degeneracy,ass.clt_K.stability} reduce to \Cref{ass.clt_2cf.moments,ass.clt_2cf.degeneracy,ass.clt_2cf.stability}.
For \cref{ass.clt_K.locality}, $\bar r_n$ is the average of the fold means $R_k=m_n^{-1}\sum_{i\in\s_k}r_i$, $k=1,2$, which have equal variances by fold symmetry.
Minkowski's inequality and \Cref{lem.clt_2cf.stability_bounds} then give
$$
\Var{\bar r_n}
\le
\Var{R_1}
=
\frac{2}{n}\Var{r_1}+\bp{1-\frac2n}\Cov{r_1,r_2}
\le
\frac{\gamma_n}{2}+\abs{\Cov{r_1,r_2}},
$$
so \Cref{ass.clt_2cf.degeneracy,ass.clt_2cf.stability,ass.clt_2cf.locality} imply \cref{ass.clt_K.locality}.
Thus \Cref{thm.clt_K} contains \Cref{thm.clt_2cf} as the special case $K_n=2$.

\begin{proof}
For $i=1,\dots,n$, write $f_i=\tilde f_{\etahat_{-k(i)}}(W_i)$, so that $T_n=(nV_n)^{-1/2}\sum_{i=1}^n f_i$.
Since $\etahat_{-k(i)}$ is $\Dc_{-i}$-measurable, \cref{eq.clt_K.normalization,eq.clt_K.v} give
\begin{equation}
\E{f_i\mid\Dc_{-i}}=0,
\qquad
\E{f_i^2\mid\Dc_{-i}}=v_{k(i),n}.
\label{eq.clt_K.proof.conditional}
\end{equation}
When $p<\infty$, the conditional Jensen inequality gives $v_{k,n}^{p/2}\le\E{\abs{\tilde f_{\etahat_{-k}}(W)}^p\mid\etahat_{-k}}$, and when $p=\infty$, $v_{k,n}\le\bar c^2$ directly.
Hence \cref{ass.clt_K.moments} and fold symmetry give, for all $i$ and $k$,
\begin{equation}
\norm{f_i}_p\le\bar c,
\qquad
\norm{v_{k,n}^{1/2}}_p\le\bar c.
\label{eq.clt_K.proof.moments}
\end{equation}
For $i=1,\dots,n$, let $\Delta_i=T_n-\E{T_n\mid\Dc_{-i}}$ and, for $\ell\ne k(i)$, let $A_{i,\ell}=S_{\ell,n}-\E{S_{\ell,n}\mid\Dc_{-i}}$.
Let $q=1/(1-1/p)\in[1,2)$ denote the conjugate exponent of $p$.

\paragraph{Step 1. Increment bound.}
Conditional on $\Dc_{-i}$, the scores $f_j$, $j\in\s_{k(i)}\setminus\bc{i}$, are fixed, while $f_i$ and the models $\etahat_{-\ell}$, $\ell\ne k(i)$, vary.
Hence, by \cref{eq.clt_K.proof.conditional},
\begin{equation}
\Delta_i
=
\frac{f_i}{\sqrt{nV_n}}
+
\frac{1}{\sqrt{K_nV_n}}
\sum_{\ell\ne k(i)}A_{i,\ell}.
\label{eq.clt_K.proof.delta}
\end{equation}
Let $\etahat_{-\ell}'$ denote the model obtained after replacing $W_i$ by an independent copy in the training sample of $\etahat_{-\ell}$, and let $S_{\ell,n}'$ denote $S_{\ell,n}$ computed with $\etahat_{-\ell}'$ in place of $\etahat_{-\ell}$.
Conditional on $\Dc_{-i}$, $S_{\ell,n}$ and $S_{\ell,n}'$ are independent and identically distributed, so $\E{A_{i,\ell}^2}=\frac12\E{\bp{S_{\ell,n}-S_{\ell,n}'}^2}$.
The observations in $\s_\ell$ are independent of $(\etahat_{-\ell},\etahat_{-\ell}')$, so, conditional on these models, the differences $\tilde f_{\etahat_{-\ell}}(W_j)-\tilde f_{\etahat_{-\ell}'}(W_j)$, $j\in\s_\ell$, are i.i.d.\ with mean zero, since $P\tilde f_\eta=0$ for every $\eta\in\Hc$.
Hence the cross terms of the square vanish, and fold symmetry and exchangeability within folds give $\E{\bp{S_{\ell,n}-S_{\ell,n}'}^2}=\gamma_n$, so $\E{A_{i,\ell}^2}=\gamma_n/2$.
By \cref{eq.clt_K.proof.delta}, the Cauchy--Schwarz inequality, $\E{f_i^2}=1$, $(K_n-1)^2/K_n\le K_n-1$, and \cref{ass.clt_K.degeneracy},
\begin{equation}
\E{\Delta_i^2}
\le
\frac{2}{nV_n}
+
\frac{(K_n-1)^2}{K_nV_n}\gamma_n
\le
\frac{2}{\ubar c}
\bc{\frac1n+(K_n-1)\gamma_n}.
\label{eq.clt_K.proof.delta_bound}
\end{equation}

\paragraph{Step 2. Kernel concentration.}
For $i=1,\dots,n$, let
$$
c_i
=
\Cov{f_i,\sum_{j=1}^n f_j\,\middle|\,\Dc_{-i}}
=
\sqrt{nV_n}\,\E{f_i\Delta_i\mid\Dc_{-i}},
$$
where the second equality uses $\E{f_i\mid\Dc_{-i}}=0$.
The term $j=i$ contributes $v_{k(i),n}$ by \cref{eq.clt_K.proof.conditional}, the terms $j\in\s_{k(i)}\setminus\bc{i}$ are $\Dc_{-i}$-measurable and contribute zero, and $\sum_{j\in\s_\ell}f_j=\sqrt{m_n}\,S_{\ell,n}$ for $\ell\ne k(i)$, so \cref{eq.clt_K.r} gives
\begin{equation}
c_i=v_{k(i),n}+(K_n-1)r_i.
\label{eq.clt_K.proof.c}
\end{equation}
Since $\sum_{i=1}^n v_{k(i),n}=m_n\sum_{k=1}^{K_n}v_{k,n}$, \cref{eq.clt_K.kernel} yields $\Kc_n=(nV_n)^{-1}\sum_{i=1}^n c_i$.
Since $\E{\Kc_n}=1$ and $\Var{X+Y}\le2\Var{X}+2\Var{Y}$,
\begin{equation}
\E{\abs{\Kc_n-1}}^2
\le
\Var{\Kc_n}
\le
\frac{2\Var{\bar v_n}}{V_n^2}
+
\frac{2(K_n-1)^2\Var{\bar r_n}}{V_n^2},
\label{eq.clt_K.proof.kernel}
\end{equation}
which converges to zero by \cref{ass.clt_K.scale,ass.clt_K.locality}.

\paragraph{Step 3. Stein remainder.}
Fix $h\in\Lip{1}$, and let $\varphi_h$ be the solution to \cref{eq.stein_eq}, so that $\norm{\varphi_h'}_\infty\le1$ and $\varphi_h'$ is $2$-Lipschitz.
As in \cref{thm.clt_2cf.stein_bound.reduction},
\begin{equation}
d_W\of{T_n,\Nc\of{0,1}}
\le
\E{\abs{\Kc_n-1}}
+
\sup_{h\in\Lip{1}}\abs{R_n(\varphi_h)},
\qquad
R_n(\varphi_h)=\E{T_n\varphi_h(T_n)}-\E{\Kc_n\varphi_h'(T_n)}.
\label{eq.clt_K.proof.reduction}
\end{equation}
For each $i$, since $\E{f_i\mid\Dc_{-i}}=0$ and $T_n-\Delta_i=\E{T_n\mid\Dc_{-i}}$ is $\Dc_{-i}$-measurable,
\begin{align*}
\E{f_i\varphi_h(T_n)}
&=
\E{f_i\bc{\varphi_h(T_n)-\varphi_h(T_n-\Delta_i)}}\\
&=
\E{f_i\Delta_i\varphi_h'(T_n-\Delta_i)}
+\varepsilon_{i,1}(h)\\
&=
\frac{1}{\sqrt{nV_n}}\E{c_i\varphi_h'(T_n-\Delta_i)}
+\varepsilon_{i,1}(h)\\
&=
\frac{1}{\sqrt{nV_n}}\E{c_i\varphi_h'(T_n)}
+\varepsilon_{i,1}(h)+\varepsilon_{i,2}(h),
\end{align*}
where the second and fourth equalities define the remainders $\varepsilon_{i,1}(h)$ and $\varepsilon_{i,2}(h)$.
Multiplying by $(nV_n)^{-1/2}$, summing over $i$, and using $\Kc_n=(nV_n)^{-1}\sum_{i=1}^n c_i$ gives
$$
\abs{R_n(\varphi_h)}
\le
\frac{1}{\sqrt{nV_n}}
\sum_{i=1}^n
\bc{\abs{\varepsilon_{i,1}(h)}+\abs{\varepsilon_{i,2}(h)}}.
$$

As in \cref{thm.clt_2cf.step2.taylor_bound}, the Taylor remainder satisfies
$$
\abs{\varepsilon_{i,1}(h)}
\le
\E{\abs{f_i}\min\{\Delta_i^2,2\abs{\Delta_i}\}}
\le
2\E{\abs{f_i}\abs{\Delta_i}\min\{\abs{\Delta_i},1\}}
\le
2\bar c\,\E{\Delta_i^2}^{1-1/p},
$$
where the last inequality uses H\"older's inequality with conjugate exponents $p$ and $q$, \cref{eq.clt_K.proof.moments}, and $\bp{\abs{d}\min\{\abs{d},1\}}^q\le d^2$ for $d\in\Rb$.

For the shift remainder, $\abs{\varphi_h'(x)-\varphi_h'(y)}\le2\min\{\abs{x-y},1\}$ and the $\Dc_{-i}$-measurability of $c_i$ give
$$
\abs{\varepsilon_{i,2}(h)}
\le
\frac{2}{\sqrt{nV_n}}
\E{\abs{c_i}\,\E{\min\{\abs{\Delta_i},1\}\mid\Dc_{-i}}}.
$$
Let $Y_i=\E{\Delta_i^2\mid\Dc_{-i}}$.
By the conditional Cauchy--Schwarz inequality and \cref{eq.clt_K.proof.conditional}, $\abs{c_i}\le\sqrt{nV_n}\,v_{k(i),n}^{1/2}Y_i^{1/2}$ and $\E{\min\{\abs{\Delta_i},1\}\mid\Dc_{-i}}\le\min\{Y_i,1\}^{1/2}$, so
$$
\abs{\varepsilon_{i,2}(h)}
\le
2\,\E{v_{k(i),n}^{1/2}G_i},
\qquad
G_i=\bp{Y_i\min\{Y_i,1\}}^{1/2}.
$$
Since $q\in[1,2)$, $G_i^q=Y_i^q\le Y_i$ when $Y_i\le1$, and $G_i^q=Y_i^{q/2}\le Y_i$ when $Y_i>1$.
H\"older's inequality with conjugate exponents $p$ and $q$ and \cref{eq.clt_K.proof.moments} then give
$$
\abs{\varepsilon_{i,2}(h)}
\le
2\norm{v_{k(i),n}^{1/2}}_p\E{Y_i}^{1-1/p}
\le
2\bar c\,\E{\Delta_i^2}^{1-1/p}.
$$
Both bounds are uniform over $h\in\Lip{1}$, so \cref{ass.clt_K.degeneracy,eq.clt_K.proof.delta_bound} give
\begin{equation}
\sup_{h\in\Lip{1}}\abs{R_n(\varphi_h)}
\le
4\bar c\sqrt{\frac{n}{V_n}}\,
\max_{i}\E{\Delta_i^2}^{1-1/p}
\le
C\sqrt n\bc{\frac1n+(K_n-1)\gamma_n}^{1-1/p},
\label{eq.clt_K.proof.stein}
\end{equation}
where $C=4\bar c\,\ubar c^{-1/2}(2/\ubar c)^{1-1/p}$.

\paragraph{Step 4. Conclusion.}
By \cref{eq.clt_K.proof.reduction,eq.clt_K.proof.kernel,eq.clt_K.proof.stein}, it remains to show that the right-hand side of \cref{eq.clt_K.proof.stein} vanishes.
Since $(a+b)^s\le a^s+b^s$ for $a,b\ge0$ and $s\in(0,1]$,
$$
\sqrt n\bc{\frac1n+(K_n-1)\gamma_n}^{1-1/p}
\le
n^{-\bp{\frac12-\frac1p}}
+
\bc{n^{\frac{1}{2(1-1/p)}}(K_n-1)\gamma_n}^{1-1/p},
$$
which converges to zero because $p>2$ and by \cref{ass.clt_K.stability}.
Hence $d_W\of{T_n,\Nc\of{0,1}}\to0$, which implies $T_n\leadsto\Nc\of{0,1}$.
\end{proof}

\section{Proofs}
\label{app.proofs}

Throughout this section, let $k(i)$ denote the fold not containing observation $i$, $i \in \stilde_{k(i)}$, so that $W_i$ evaluates the model $\etahat_{k(i)}$ trained on the other fold, and define 
$$f_i = f_{\etahat_{k(i)}}(W_i), \qquad \Delta_i = T_n - \E{T_n \mid \Dc_{-i}}, \qquad i = 1, \dots, n,$$
so that 
$$T_n = \frac{\sum_{i=1}^n f_i}{\sqrt{n \bp{1 + \rho_n}}}, \qquad \E{T_n} = 0, \qquad \Var{T_n} = 1.$$
Define also the approximate Stein kernel
$$\Kc_n = \frac{1 + \frac{2}{n} \sum_{i \in \s_1} r_i}{1 + \rho_n}.$$
For a function $g \colon \Rb \to \Rb$, let $\norm{g}_\infty = \sup_{x \in \Rb} \abs{g(x)}$. 

\subsection{Auxiliary Results}

\begin{lemma}[Stability Bounds]
\label{lem.clt_2cf.stability_bounds}
    It holds that 
    \begin{equation} \label{lem.clt_2cf.stability_bounds.variance}
        \E{\Var{\thetahat_{\etahat_1} \,\middle|\, \Dc_{-1}}} = \frac{\gamma_n}{n},
    \end{equation}
    and 
    \begin{equation} \label{lem.clt_2cf.stability_bounds.r}
        \Var{r_1} \le \E{r_1^2} \le \frac{n \gamma_n}{4}.
    \end{equation}
    Moreover, 
    \begin{equation} \label{lem.clt_2cf.stability_bounds.rho}
        \E{r_1} = \rho_n, 
        \qquad \text{and hence} \qquad 
        \abs{\rho_n} \le \frac{1}{2} \sqrt{n \gamma_n}.
    \end{equation}
\end{lemma}

\begin{proof}
Let
$$
\thetahat_{\etahat'_1}
=
\frac2n\sum_{j\in\s_2}f_{\etahat'_1}(W_j)
$$
denote the estimator $\thetahat_{\etahat_1}$ computed with $\etahat'_1$ in place of $\etahat_1$.
Then
\begin{align}
\E{\Var{\thetahat_{\etahat_1}\mid\Dc_{-1}}}
&=
\frac12
\E{
\E{
\bp{\thetahat_{\etahat_1}-\thetahat_{\etahat'_1}}^2
\mid\Dc_{-1}
}
}
\label{lem.clt_2cf.stability_bounds.replacement}\\
&=
\frac12
\E{\bp{\thetahat_{\etahat_1}-\thetahat_{\etahat'_1}}^2}
\nonumber\\
&=
\frac{2}{n^2}
\E{
\bp{
\sum_{j\in\s_2}
\bc{f_{\etahat_1}(W_j)-f_{\etahat'_1}(W_j)}
}^2
}
\nonumber\\
&=
\frac{2}{n^2}
\sum_{j\in\s_2}
\E{\bp{f_{\etahat_1}(W_j)-f_{\etahat'_1}(W_j)}^2}
\label{lem.clt_2cf.stability_bounds.evaluation}\\
&=
\frac{2}{n^2}\cdot\frac n2\,\gamma_n
=
\frac{\gamma_n}{n}.
\nonumber
\end{align}
\cref{lem.clt_2cf.stability_bounds.replacement} uses that, conditional on $\Dc_{-1}$, $\thetahat_{\etahat_1}$ and $\thetahat_{\etahat'_1}$ are independent and identically distributed, so that $\Var{X \mid \Dc_{-1}} = \frac12 \E{(X - X')^2 \mid \Dc_{-1}}$.
\cref{lem.clt_2cf.stability_bounds.evaluation} expands the square; the cross terms vanish because, conditional on $(\etahat_1, \etahat'_1)$, the differences $f_{\etahat_1}(W_j)-f_{\etahat'_1}(W_j)$ are independent across $j \in \s_2$ with mean $P f_{\etahat_1} - P f_{\etahat'_1} = 0$.
The last equality uses that $W_j$, $j \in \s_2$, is independent of $(\etahat_1, \etahat'_1)$ and $\abs{\s_2} = n/2$.

For \cref{lem.clt_2cf.stability_bounds.r}, the first inequality is $\Var{X} \le \E{X^2}$, and
\begin{align}
\E{r_1^2}
&\le
\frac{n^2}{4}\E{
\Var{\sqrt{n/2}\,\thetahat_{\etahat_1}\mid\Dc_{-1}}
\Var{\sqrt{n/2}\,\thetahat_{\etahat_2}\mid\Dc_{-1}}
}
\label{lem.clt_2cf.stability_bounds.cs}\\
&=
\frac{n^2}{4}\E{\Var{\thetahat_{\etahat_1}\mid\Dc_{-1}}}
\label{lem.clt_2cf.stability_bounds.cond_var}\\
&=
\frac{n\gamma_n}{4}.
\nonumber
\end{align}
\cref{lem.clt_2cf.stability_bounds.cs} is conditional Cauchy--Schwarz.
\cref{lem.clt_2cf.stability_bounds.cond_var} uses that $\etahat_2$ and $W_i$, $i \in \s_1 \setminus \bc{1}$, are $\Dc_{-1}$-measurable, so that only $f_{\etahat_2}(W_1)$ varies and, by the standardization,
$$
\Var{\sqrt{n/2}\,\thetahat_{\etahat_2}\mid\Dc_{-1}}
=
\frac n2\bp{\frac2n}^2\Var{f_{\etahat_2}(W_1)\mid\Dc_{-1}}
=
\frac2n.
$$
The last equality is \cref{lem.clt_2cf.stability_bounds.variance}.

For \cref{lem.clt_2cf.stability_bounds.rho}, 
\begin{align}
    \E{r_1} & = \frac{n^2}{4} \E{\Cov{\thetahat_{\etahat_1}, \thetahat_{\etahat_2} \;\middle|\; \Dc_{-1} }} \nonumber \\
    & = \frac{n}{2} \E{\Cov{f_{\etahat_2}(W_1), \thetahat_{\etahat_1} \;\middle|\; \Dc_{-1} }} \nonumber \\
    & = \frac{n}{2} \E{f_{\etahat_2}(W_1) \thetahat_{\etahat_1} } \label{lem.clt_2cf.stability_bounds.cross_moment} \\
    & = \frac{n}{2} \Cov{\thetahat_{\etahat_1}, \thetahat_{\etahat_2} } \nonumber \\
    & = \rho_n. \nonumber 
\end{align}
The second equality uses that $\etahat_2$ and $W_i$, $i \in \s_1 \setminus \bc{1}$, are $\Dc_{-1}$-measurable. 
\cref{lem.clt_2cf.stability_bounds.cross_moment} uses $\E{f_{\etahat_2}(W_1) \mid \Dc_{-1}} = P f_{\etahat_2} = 0$. 
The fourth equality uses exchangeability within $\s_1$ and $\E{\thetahat_{\etahat_2}} = 0$, and the last uses $\Var{\thetahat_{\etahat_k}} = 2/n$, which follows from the standardization. 
Finally, by Jensen's inequality and \cref{lem.clt_2cf.stability_bounds.r}, $\abs{\rho_n} = \abs{\E{r_1}} \le \sqrt{\E{r_1^2}} \le \frac{1}{2} \sqrt{n \gamma_n}$.
\end{proof}

\subsection{Proof of \Cref{lem.clt_2cf.sufficient}}

\begin{proof}
By exchangeability within $\s_1$, $\Var{r_1} = \Var{r_2}$, so $\Cov{r_1, r_2} = \Var{r_1} \Corr{r_1, r_2}$, where both sides are zero if $\Var{r_1} = 0$.
Hence, by \cref{lem.clt_2cf.stability_bounds.r} and $\abs{\Corr{r_1, r_2}} \le 1$,
$$\abs{\Cov{r_1, r_2}} \le \frac{n \gamma_n}{4} \abs{\Corr{r_1, r_2}} \le \frac{n \gamma_n}{4}.$$
Under (a), $n^{\frac{1}{2(1-1/p)}} \gamma_n = n^{-\frac{1-2/p}{2(1-1/p)}} \cdot n \gamma_n \to 0$ since $p > 2$, and the middle term is $O(1) \cdot o(1)$.
Under (b), the first condition is \Cref{ass.clt_2cf.stability} with $p = \infty$, and the middle term equals $\frac{1}{4} \sqrt{n} \gamma_n \cdot \sqrt{n} \abs{\Corr{r_1, r_2}} = o(1) \cdot O(1)$.
Under (c), $n^{\frac{1}{2(1-1/p)}} \gamma_n \le n \gamma_n \to 0$ since $\frac{1}{2(1-1/p)} \le 1$, and the right-hand side is $o(1)$.
\end{proof}

\subsection{A Stein Bound}

The following theorem is the common core of the proofs of \Cref{thm.clt_2cf,thm.clt_2cf.quantitative}. 
It uses only the setup of \Cref{sec.clt_2cf}, not \Cref{ass.clt_2cf}. 

\begin{theorem}[Stein Bound]
\label{thm.clt_2cf.stein_bound}
For every $n$, 
\begin{equation} \label{thm.clt_2cf.stein_bound.mean}
    \E{\Kc_n} = 1,
\end{equation}
\begin{equation} \label{thm.clt_2cf.stein_bound.kernel}
    \Var{\Kc_n} = \frac{2}{(1 + \rho_n)^2} \frac{\Var{r_1}}{n} + \frac{1}{(1 + \rho_n)^2} \bp{1 - \frac{2}{n}} \Cov{r_1, r_2},
\end{equation}
\begin{equation} \label{thm.clt_2cf.stein_bound.increment}
    \E{\Delta_1^2} = \frac{1 + 2\rho_n + \frac{n}{4} \gamma_n}{n (1 + \rho_n)},
\end{equation}
and 
\begin{equation} \label{thm.clt_2cf.stein_bound.bound}
    d_W\of{T_n, \Nc\of{0,1}} \le \sqrt{\Var{\Kc_n}} + \sqrt{\frac{n}{1+\rho_n}} \, \bs{\E{\abs{f_1} \min\{\Delta_1^2, 2 \abs{\Delta_1}\}} + 2 \E{\Delta_1^2}}.
\end{equation}
\end{theorem}

\begin{proof}

The proof uses Stein's method of normal approximation \citep{steinBoundErrorNormal1972,steinApproximateComputationExpectations1986}. 
Its starting point is that a random variable $Z$ follows a standard normal distribution, $Z \sim \Nc \of{0,1}$, if, and only if,
$$\E{Z \varphi(Z)} = \E{\varphi'(Z)}$$
for every $\varphi$ in a sufficiently rich class. 
This characterization suggests that if there is a random variable $\Kc_n$, often called an approximate Stein kernel, satisfying two properties, 
$$\E{T_n \varphi(T_n)} - \E{\Kc_n \varphi'(T_n)} \approx 0,$$
and 
$$\Kc_n \Lto 1,$$
then $T_n$ converges to a normal distribution. 
The main novelty of this proof is in constructing a specific $\Kc_n$ that satisfies both requirements under the locality and regularity conditions. 

First, for any $h \in \Lip{1}$, let $\varphi_h$ be the solution to Stein's equation 
\begin{equation} \label{eq.stein_eq}
    h(t) - \E{h(Z)} = \varphi'_h(t) - t \varphi_h(t), \qquad t \in \Rb,
\end{equation}
given by 
$$\varphi_h(t) = e^{t^2/2} \int_{-\infty}^t \bs{h(x) - \E{h(Z)}} e^{-x^2/2} \, \rmd x.$$
In particular, 
$\norm{\varphi_h'}_\infty \le 1$ and $\absx{\varphi_h'(x) - \varphi_h'(y)} \le 2 \absx{x - y}$ for all $x, y \in \Rb$ \citep[Lemma 4.2]{chatterjeeNewMethodNormal2008a}.
Second, define the Stein remainder
\begin{equation} \label{eq.remainder}
    R_n(\varphi_h) = \E{T_n \varphi_h (T_n)} - \E{\Kc_n \varphi_h'(T_n)}.
\end{equation} 
Combining \cref{eq.stein_eq}, \cref{eq.remainder}, and $\norm{\varphi_h'}_\infty \le 1$ yields the inequality 
\begin{align*}
    \abs{\E{h(T_n)} - \E{h(Z)}} & = \abs{\E{T_n \varphi_h(T_n)} - \E{\varphi'_h(T_n)}} \\ %
    & = \abs{\E{(1-\Kc_n) \varphi'_h(T_n)} - R_n(\varphi_h)} \\
    & \le \E{\abs{\Kc_n - 1}} + \abs{R_n(\varphi_h)}.
\end{align*}
Hence, 
\begin{equation} \label{thm.clt_2cf.stein_bound.reduction}
    d_W\of{T_n, \Nc\of{0,1}} \le \E{\abs{\Kc_n - 1}} + \sup_{h \in \Lip{1}} \abs{R_n(\varphi_h)}.
\end{equation}
The proof proceeds in two steps. 
Step 1 establishes \cref{thm.clt_2cf.stein_bound.mean,thm.clt_2cf.stein_bound.kernel} and bounds the first term in \cref{thm.clt_2cf.stein_bound.reduction}. 
Step 2 bounds the second term and establishes \cref{thm.clt_2cf.stein_bound.increment}. 

\paragraph{Step 1. Kernel Concentration.}

\begin{align}
    \E{\abs{\Kc_n - 1}}^2 & \le \Var{\Kc_n} \label{thm.clt_2cf.step1.1} \\
    & = \frac{4}{n^2 (1 + \rho_n)^2} \Var{\sum_{i \in \s_1} r_i} \nonumber \\
    & = \frac{4}{n^2 (1 + \rho_n)^2} \bp{\frac{n}{2} \Var{r_1} + \frac{n}{2} \bp{\frac{n}{2} - 1} \Cov{r_1, r_2}} \label{thm.clt_2cf.step1.2} \\
    & = \frac{2}{(1 + \rho_n)^2} \frac{\Var{r_1}}{n} + \frac{1}{(1 + \rho_n)^2} \bp{1 - \frac{2}{n}} \Cov{r_1, r_2}. \nonumber 
\end{align}

\cref{thm.clt_2cf.step1.1} follows since 
$$\E{\Kc_n} = (1 + \rho_n)^{-1} \bp{1 + \E{r_1}} = 1,$$
by exchangeability of observations within $\s_1$ and \cref{lem.clt_2cf.stability_bounds.rho}.
\cref{thm.clt_2cf.step1.2} uses exchangeability within $\s_1$. 
This proves \cref{thm.clt_2cf.stein_bound.mean,thm.clt_2cf.stein_bound.kernel}. 

\paragraph{Step 2. Vanishing Stein Remainder.}
Fix $h \in \Lip{1}$.
The key identity is 
\begin{align}
\E{T_n\varphi_h(T_n)}
&=
\sqrt{\frac{n}{1+\rho_n}}\,
\E{f_1\varphi_h(T_n)}
\label{thm.clt_2cf.step2.symmetry}\\
&=
\sqrt{\frac{n}{1+\rho_n}}\,
\E{f_1\bc{\varphi_h(T_n)-\varphi_h(T_n-\Delta_1)}}
\label{thm.clt_2cf.step2.centering}\\
&=
\sqrt{\frac{n}{1+\rho_n}}\,
\E{f_1\Delta_1\varphi_h'(T_n-\Delta_1)}
+\varepsilon_{n,1}(h)
\label{thm.clt_2cf.step2.linearization}\\
&=
\frac{1}{1+\rho_n}
\E{\bp{1+r_1}\varphi_h'(T_n-\Delta_1)}
+\varepsilon_{n,1}(h)
\label{thm.clt_2cf.step2.conditioning}\\
&=
\frac{1}{1+\rho_n}
\E{\bp{1+r_1}\varphi_h'(T_n)}
+\varepsilon_{n,1}(h)+\varepsilon_{n,2}(h)
\label{thm.clt_2cf.step2.shift}\\
&=
\E{\Kc_n\varphi_h'(T_n)}
+\varepsilon_{n,1}(h)+\varepsilon_{n,2}(h),
\label{thm.clt_2cf.step2.kernel}
\end{align}
where the deterministic errors $\varepsilon_{n,1}(h)$ and $\varepsilon_{n,2}(h)$ are defined by \cref{thm.clt_2cf.step2.linearization,thm.clt_2cf.step2.shift}, respectively.
Thus,
$$R_n(\varphi_h)=\varepsilon_{n,1}(h)+\varepsilon_{n,2}(h).$$

\cref{thm.clt_2cf.step2.symmetry} follows from fold symmetry and exchangeability within $\s_1$.
\cref{thm.clt_2cf.step2.centering} subtracts 
$$
\E{f_1\varphi_h(T_n-\Delta_1)}
=
\E{\E{f_1\mid\Dc_{-1}} \varphi_h(T_n-\Delta_1)}
=0.
$$
\cref{thm.clt_2cf.step2.linearization} defines the Taylor expansion remainder $\varepsilon_{n,1}(h)$. 
\cref{thm.clt_2cf.step2.conditioning} uses %
\begin{align}
\sqrt{n(1+\rho_n)}\,\E{f_1\Delta_1\,\middle|\,\Dc_{-1}}
&=
\E{
f_1\bc{
\sum_{j=1}^n f_j
-\E{\sum_{j=1}^n f_j\,\middle|\,\Dc_{-1}}
}
\,\middle|\,\Dc_{-1}
}
\nonumber\\
&=
\Cov{f_1,\sum_{j=1}^n f_j\,\middle|\,\Dc_{-1}}
\nonumber\\
&=
1+\frac n2
\Cov{f_1,\thetahat_{\etahat_1}\,\middle|\,\Dc_{-1}}
\nonumber\\
&=
1+r_1.
\label{thm.clt_2cf.step2.coefficient}
\end{align}
\cref{thm.clt_2cf.step2.shift} defines the remainder $\varepsilon_{n,2}(h)$, and \cref{thm.clt_2cf.step2.kernel} follows from exchangeability within $\s_1$ and the definition of $\Kc_n$.

To bound $\varepsilon_{n,1}(h)$, note that, for all $t,d\in\Rb$, 
\begin{align*}
\abs{\varphi_h(t)-\varphi_h(t-d)-d\varphi_h'(t-d)}
&=
\abs{d}\,
\abs{
\int_0^1
\bs{\varphi_h'(t-d+sd)-\varphi_h'(t-d)}
\,\rmd s
}\\
&\le
\abs{d}\int_0^1
\min\{2s\abs{d},2\}\,\rmd s\\
&\le
\min\{d^2,2\abs{d}\}\\
&\le
2\abs{d}\min\{\abs{d},1\},
\end{align*}
where the first inequality uses the $2$-Lipschitz property of $\varphi_h'$ and the bound $\norm{\varphi_h'}_\infty\le1$.
Consequently,
\begin{align}
\abs{\varepsilon_{n,1}(h)}
&=
\sqrt{\frac{n}{1+\rho_n}}\,
\abs{
\E{f_1\bc{
\varphi_h(T_n)-\varphi_h(T_n-\Delta_1)
-\Delta_1\varphi_h'(T_n-\Delta_1)
}}
}
\nonumber\\
&\le
\sqrt{\frac{n}{1+\rho_n}}\,
\E{\abs{f_1}\min\{\Delta_1^2,2\abs{\Delta_1}\}}
\label{thm.clt_2cf.step2.taylor_bound}
\end{align}

For the remainder $\varepsilon_{n,2}(h)$, since $\varphi_h'$ is $2$-Lipschitz,
\begin{align}
\abs{\varepsilon_{n,2}(h)}
&=
\frac{1}{1+\rho_n}
\abs{
\E{\bp{1+r_1}
\bc{\varphi_h'(T_n-\Delta_1)-\varphi_h'(T_n)}}
}
\nonumber\\
&\le
\frac{2}{1+\rho_n}
\E{\abs{1+r_1}\abs{\Delta_1}}
\nonumber\\
&\le
2\sqrt{\frac{n}{1+\rho_n}}\,
\E{\Delta_1^2}.
\label{thm.clt_2cf.step2.shift_bound}
\end{align}
The last inequality follows from Cauchy--Schwarz, since \cref{thm.clt_2cf.step2.coefficient} and conditional Cauchy--Schwarz give
\begin{align*}
\E{\bp{1+r_1}^2}
&=
n(1+\rho_n)\,
\E{\E{f_1\Delta_1\mid\Dc_{-1}}^2}\\
&\le
n(1+\rho_n)\,
\E{
\E{f_1^2\mid\Dc_{-1}}
\E{\Delta_1^2\mid\Dc_{-1}}
}\\
&=
n(1+\rho_n)\E{\Delta_1^2},
\end{align*}
where $\E{f_1^2\mid\Dc_{-1}}=1$ by the standardization.

Since both error bounds are uniform over $h\in\Lip{1}$, combining \cref{thm.clt_2cf.step2.taylor_bound,thm.clt_2cf.step2.shift_bound} gives
\begin{align}
\sup_{h\in\Lip{1}}\abs{R_n(\varphi_h)}
&=
\sup_{h\in\Lip{1}}
\abs{\varepsilon_{n,1}(h)+\varepsilon_{n,2}(h)}
\nonumber\\
&\le
\sup_{h\in\Lip{1}}\abs{\varepsilon_{n,1}(h)}
+
\sup_{h\in\Lip{1}}\abs{\varepsilon_{n,2}(h)}
\nonumber\\
&\le
\sqrt{\frac{n}{1+\rho_n}}\,
\bs{
\E{\abs{f_1}\min\{\Delta_1^2,2\abs{\Delta_1}\}}
+
2\E{\Delta_1^2}
}.
\label{thm.clt_2cf.step2.remainder_bound}
\end{align}
Together with \cref{thm.clt_2cf.step1.1} and \cref{thm.clt_2cf.stein_bound.reduction}, this proves \cref{thm.clt_2cf.stein_bound.bound}. 

It remains to establish \cref{thm.clt_2cf.stein_bound.increment}.
Since the scores of all other fold-1 observations are $\Dc_{-1}$-measurable and $\E{f_1\mid\Dc_{-1}}=0$,
$$
\Delta_1
=
\frac{
f_1+\frac n2
\bc{
\thetahat_{\etahat_1}
-\E{\thetahat_{\etahat_1}\mid\Dc_{-1}}
}
}{
\sqrt{n(1+\rho_n)}
}.
$$
Expanding the square,
\begin{align*}
\E{\Delta_1^2}
&=
\frac{1}{n(1+\rho_n)}
\bs{
\E{f_1^2}
+ n \E{f_1 \bc{\thetahat_{\etahat_1} - \E{\thetahat_{\etahat_1}\mid\Dc_{-1}}}}
+ \frac{n^2}{4} \E{\Var{\thetahat_{\etahat_1}\mid\Dc_{-1}}}
}
\\
&=
\frac{1 + 2\rho_n + \frac{n}{4}\gamma_n}{n(1+\rho_n)}.
\end{align*}
The second equality uses $\E{f_1^2} = 1$, the identity $\E{f_1 \thetahat_{\etahat_1}} = 2\rho_n/n$ from \cref{lem.clt_2cf.stability_bounds.cross_moment}, $\E{f_1 \E{\thetahat_{\etahat_1}\mid\Dc_{-1}}} = \E{\E{f_1\mid\Dc_{-1}} \E{\thetahat_{\etahat_1}\mid\Dc_{-1}}} = 0$, and \cref{lem.clt_2cf.stability_bounds.variance}.
\end{proof}

\subsection{Proof of \Cref{thm.clt_2cf}}
\label{app.proofs.clt_2cf}

\begin{proof}
The main steps of the proof are first established in \Cref{thm.clt_2cf.stein_bound}. 
After that, it suffices to show that $\Var{\Kc_n} \to 0$ and that $\sqrt{n} \, \bs{\E{\abs{f_1} \min\{\Delta_1^2, 2 \abs{\Delta_1}\}} + 2 \E{\Delta_1^2}} \to 0$. 

For the first claim, \cref{thm.clt_2cf.stein_bound.kernel} and \cref{lem.clt_2cf.stability_bounds.r} give
$$
\Var{\Kc_n} \le \frac{\gamma_n}{2 (1 + \rho_n)^2} + \frac{\abs{\Cov{r_1, r_2}}}{(1 + \rho_n)^2} \longrightarrow 0,
$$
since $1 + \rho_n$ is bounded away from zero by \Cref{ass.clt_2cf.degeneracy}, $\gamma_n \to 0$ by \Cref{ass.clt_2cf.stability}, and $\Cov{r_1, r_2} \to 0$ by \Cref{ass.clt_2cf.locality}.

For the second claim, first note that \cref{thm.clt_2cf.stein_bound.increment}, $\rho_n \le 1$, and \Cref{ass.clt_2cf.degeneracy} give
\begin{equation} \label{thm.clt_2cf.step2.increment_bound}
    \E{\Delta_1^2} \le \frac{2}{n} + \frac{\gamma_n}{4 \ubar{c}} \le C_0 \bp{\frac{1}{n} + \gamma_n}, \qquad C_0 = \max\bc{2, \frac{1}{4 \ubar{c}}}.
\end{equation}
Next, let $q = 1/(1-1/p) \in [1, 2)$ denote the conjugate exponent of $p$.
By H\"older's inequality,
\begin{align*}
\E{\abs{f_1}\min\{\Delta_1^2,2\abs{\Delta_1}\}}
&\le
2\E{\abs{f_1}\abs{\Delta_1}\min\{\abs{\Delta_1},1\}}\\
&\le
2\bar c\,
\E{\Delta_1^2}^{1-1/p},
\end{align*}
where the last inequality uses conjugate exponents $p$ and $q$, $\norm{f_1}_p = \norm{f_{\etahat_1}(W)}_p \le \bar c$ by fold symmetry and \Cref{ass.clt_2cf.moments}, and
$$
\bp{\abs{d}\min\{\abs{d},1\}}^{q}
\le d^2,
$$
which holds for all $d \in \Rb$ since $q \in [1, 2)$.
Combining these bounds gives
\begin{align*}
\sqrt{n} \, \bs{\E{\abs{f_1} \min\{\Delta_1^2, 2 \abs{\Delta_1}\}} + 2 \E{\Delta_1^2}}
&\le
2\bp{\bar c+1}
\sqrt{n}\,
\E{\Delta_1^2}^{1-1/p}\\
&\le
C\sqrt n
\bp{\frac1n+\gamma_n}^{1-1/p}\\
&=
C\bs{
n^{-\frac{1-2/p}{2(1-1/p)}}
+
n^{\frac{1}{2(1-1/p)}}\gamma_n
}^{1-1/p}
\longrightarrow0,
\end{align*}
for some constant $C<\infty$ depending only on $\bar c$, $\ubar c$, and $p$.
The first inequality uses 
$$
\E{\Delta_1^2}\le\E{T_n^2}=1
\quad\Longrightarrow\quad
\E{\Delta_1^2}
\le
\E{\Delta_1^2}^{1-1/p},
$$
where $\E{\Delta_1^2}\le\E{T_n^2}$ follows because conditional expectation is an $L^2$ projection.
The second inequality uses \cref{thm.clt_2cf.step2.increment_bound}. 
The convergence follows from $p>2$ and \Cref{ass.clt_2cf.stability}.
Since $1 + \rho_n \ge \ubar{c}$, the factor $(1+\rho_n)^{-1/2}$ in \cref{thm.clt_2cf.stein_bound.bound} is bounded, which completes the proof.
\end{proof}

\subsection{Proof of \Cref{thm.clt_2cf.quantitative}}

\begin{proof}
Let $M < \infty$ be such that $n^{1/(1-1/p)} \gamma_n \le M$ and $n \abs{\Cov{r_1, r_2}} \le M$ for all $n$.
Since $1/(1-1/p) \ge 1$, also $n \gamma_n \le M$.
Throughout, $C < \infty$ denotes a constant depending only on $\bar c$, $\ubar c$, $p$, and $M$, whose value may change from line to line.
By \cref{thm.clt_2cf.stein_bound.bound} and $1 + \rho_n \ge \ubar c$, it suffices to show that each of $\Var{\Kc_n}$, $\sqrt{n} \E{\Delta_1^2}$, and $\sqrt{n} \E{\abs{f_1} \min\{\Delta_1^2, 2\abs{\Delta_1}\}}$ is at most $C/n$, $C/\sqrt{n}$, and $C/\sqrt{n}$, respectively.

First, by \cref{thm.clt_2cf.stein_bound.kernel}, \cref{lem.clt_2cf.stability_bounds.r}, and \Cref{ass.clt_2cf.degeneracy},
$$
\Var{\Kc_n}
\le
\frac{\gamma_n}{2 \ubar c^2} + \frac{\abs{\Cov{r_1, r_2}}}{\ubar c^2}
\le
\frac{M}{2 \ubar c^2 n} + \frac{M}{\ubar c^2 n}
\le \frac{C}{n}.
$$

Second, by \cref{thm.clt_2cf.stein_bound.increment}, $\rho_n \le 1$, and \Cref{ass.clt_2cf.degeneracy},
\begin{equation} \label{thm.clt_2cf.quantitative.increment}
    \E{\Delta_1^2} \le \frac{2}{n} + \frac{\gamma_n}{4 \ubar c} \le \frac{2}{n} + \frac{M}{4 \ubar c \, n} \le \frac{C}{n}.
\end{equation}

Third, let $\phi(x) = \min\{x^2, 2\abs{x}\}$ and decompose $\Delta_1 = A + B$, where
$$
A = \frac{f_1}{\sqrt{n(1+\rho_n)}},
\qquad
B = \frac{\frac{n}{2} \bc{\thetahat_{\etahat_1} - \E{\thetahat_{\etahat_1} \mid \Dc_{-1}}}}{\sqrt{n(1+\rho_n)}}.
$$
Since $\phi$ is even, nondecreasing in $\abs{x}$, and satisfies $\phi(2x) \le 4\phi(x)$, for all $a, b \in \Rb$,
$$
\phi(a + b) \le \phi\of{2 \max\{\abs{a}, \abs{b}\}} \le 4 \max\{\phi(a), \phi(b)\} \le 4 \phi(a) + 4 \phi(b),
$$
and hence
\begin{equation} \label{thm.clt_2cf.quantitative.split}
    \E{\abs{f_1} \phi(\Delta_1)} \le 4 \E{\abs{f_1} \phi(A)} + 4 \E{\abs{f_1} \phi(B)}.
\end{equation}
For the first term in \cref{thm.clt_2cf.quantitative.split}, $\phi(A) \le A^2$ and Lyapunov's inequality with $p \ge 3$ give
$$
\E{\abs{f_1} \phi(A)}
\le
\frac{\E{\abs{f_1}^3}}{n(1+\rho_n)}
\le
\frac{\bar c^3}{\ubar c \, n}.
$$
For the second term, set $\alpha = 1 - 2/p \in [1/3, 1]$ and $q = 1/(1-1/p) \in [1, 3/2]$, so that $(1+\alpha) q = 2$ and $1/q = 1 - 1/p$.
Since $\min\{x^2, 2\abs{x}\} \le (x^2)^{\alpha} (2\abs{x})^{1-\alpha} = 2^{1-\alpha} \abs{x}^{1+\alpha}$, H\"older's inequality with conjugate exponents $p$ and $q$ gives
$$
\E{\abs{f_1} \phi(B)}
\le
2^{1-\alpha} \norm{f_1}_p \E{\abs{B}^{(1+\alpha) q}}^{1/q}
=
2^{2/p} \norm{f_1}_p \E{B^2}^{1 - 1/p}.
$$
By \cref{lem.clt_2cf.stability_bounds.variance} and \Cref{ass.clt_2cf.degeneracy},
$$
\E{B^2}
=
\frac{n}{4(1+\rho_n)} \E{\Var{\thetahat_{\etahat_1} \mid \Dc_{-1}}}
=
\frac{\gamma_n}{4(1+\rho_n)}
\le
\frac{\gamma_n}{4 \ubar c},
$$
and therefore, since $\norm{f_1}_p = \norm{f_{\etahat_1}(W)}_p \le \bar c$ by fold symmetry and $\gamma_n \le M n^{-1/(1-1/p)}$,
$$
\E{\abs{f_1} \phi(B)}
\le
2^{2/p} \bar c \bp{\frac{\gamma_n}{4 \ubar c}}^{1 - 1/p}
\le
C \gamma_n^{1 - 1/p}
\le
\frac{C}{n}.
$$
Substituting both bounds into \cref{thm.clt_2cf.quantitative.split} yields $\sqrt{n} \E{\abs{f_1} \phi(\Delta_1)} \le C / \sqrt{n}$.

Combining the three bounds with \cref{thm.clt_2cf.stein_bound.bound},
$$
d_W\of{T_n, \Nc\of{0,1}}
\le
\sqrt{\frac{C}{n}} + \frac{1}{\sqrt{\ubar c}} \bs{\frac{C}{\sqrt n} + \frac{2C}{\sqrt n}}
=
O(n^{-1/2}).
$$
\end{proof}

\subsection{Proof of \Cref{thm.inference}}
\label{app.proofs.inference}

\begin{proof}
Write $Z_k^\circ$ for $Z_{k,1}$ with its denominator replaced by $\sqrt m\,\sigma_{\etahat_{k,1}}$.
After subtracting the fitted population mean from each score, write $Z_k^\circ=U_k-V_k$, where $U_k$ and $V_k$ are the primary and reference normalized sums.
Unconditionally, the four quarters are independent i.i.d. samples.
Conditional on the two training quarters, $U_1,U_2$ are fixed, while $V_1,V_2$ are independent, centered, and have variance one.
Therefore,
\begin{equation}
\E{Z_1^\circ Z_2^\circ}=\rho_{n/2},
\qquad
\E{(Z_k^\circ)^2}=2.
\label{eq.inference.proof.oracle}
\end{equation}

For fixed values $x_1,\dots,x_{2m}$, let $Z(x)$ denote the contrast in \cref{eq.inference.contrast} under a uniform balanced split.
If $v(x)$ is the centered vector divided by its Euclidean norm $\norm{\cdot}_2$, with $v(x)=0$ for a constant vector, then
$$
Z(x)=\sqrt{\frac{2m-1}{m}}\,\epsilon^\top v(x),
$$
where $\epsilon_i$ is $1$ on the selected half and $-1$ on its complement.
The balanced-sign identities
$$
\E{\epsilon_i}=0,
\qquad
\E{\epsilon_i\epsilon_j}=-\frac1{2m-1},
\qquad
\E{\epsilon_i\epsilon_j\epsilon_\ell\epsilon_r}
=
\frac3{(2m-1)(2m-3)}
$$
for distinct indices give
\begin{equation}
\E{Z(x)}=0,
\qquad
\E{Z(x)^2}=2\ind{v(x)\ne0},
\qquad
\E{Z(x)^4}\le12.
\label{eq.inference.proof.balanced}
\end{equation}
The case $m=1$ follows directly.

Conditional on the training quarter, the pooled test observations are i.i.d.
The moment condition gives uniform integrability of the squared normalized scores, so truncation and the conditional law of large numbers imply
$$
\E{\abs{
\frac{\sigmahat_{\etahat_{k,1}}^2}
{\sigma_{\etahat_{k,1}}^2}-1
}}\to0.
$$
Conditioning on the pooled scores and randomizing their balanced split then gives
\begin{align*}
\E{(Z_{k,1}-Z_k^\circ)^2}
&\le
2\E{\bp{
\frac{\sigmahat_{\etahat_{k,1}}}
{\sigma_{\etahat_{k,1}}}-1
}^2}\\
&\le
2\E{\abs{
\frac{\sigmahat_{\etahat_{k,1}}^2}
{\sigma_{\etahat_{k,1}}^2}-1
}}
\to0.
\end{align*}
When the sample variance is zero, both contrasts are zero.
By Cauchy--Schwarz and \cref{eq.inference.proof.oracle},
$$
\E{Z_{1,1}Z_{2,1}}=\rho_{n/2}+o(1).
$$

Put $q_n=\E{Z_{1,1}Z_{2,1}\mid\Dc}$.
Applying \cref{eq.inference.proof.balanced} conditionally on the data and the training subset gives
$$
\E{Z_{k,1}^2\mid\Dc}\le2,
\qquad
\E{Z_{k,1}^4\mid\Dc}\le12,
$$
and hence $|q_n|\le2$ and $\Var{Z_{1,1}Z_{2,1}\mid\Dc}\le12$.
Conditional independence of the repetitions and the fact that projection onto $[-1,1]$ cannot increase the distance to $\rho_n$ yield
\begin{align*}
\E{(\wh\rho_n-\rho_n)^2}
&\le
\frac{12}{B}
+
\Var{q_n}
+
\bc{\E{q_n}-\rho_n}^2\\
&=
\frac{12}{B}
+
\Var{q_n}
+
\bc{\rho_{n/2}-\rho_n+o(1)}^2
\to0,
\end{align*}
where the last term vanishes because $\rho_{n/2}$ and $\rho_n$ both converge to $\rho$.

The same moment argument gives
$\sigmahat_{\etahat_k}/\sigma_{\etahat_k}\pto1$
for the original fits.
Since $\rho>-1$, $1+\rho_n$ is bounded away from zero for all large $n$.
Consequently, \cref{eq.inference.scale_balance} and correlation consistency imply
\begin{align*}
\frac{\thetahat_\etahat-\theta_\etahat}
{\wh{\se}\of{\thetahat_\etahat}}
&=
\frac{
\sqrt{n/2}\sum_{k=1}^2
(\thetahat_{\etahat_k}-\theta_{\etahat_k})/\sigma_{\etahat_k}
}{
\sqrt{2(1+\rho_n)}
}
+o_p(1)\\
&\leadsto\Nc\of{0,1},
\end{align*}
where the last step uses \Cref{thm.clt_2cf}.
Continuity of the normal distribution proves the coverage statement.
\end{proof}

\section{Inference for the BLP Example}
\label{app.inference.blp}

For the OLS estimator in \Cref{ex.setup.blp}, use the slope influence score
$$
\phi_\eta(y,x)
=
\frac{
\bc{\eta(x)-P\eta}
\bs{y-\E{Y}-\theta_\eta\bc{\eta(x)-P\eta}}
}{
P(\eta-P\eta)^2
}.
$$
Under the zero-slope null this equals the score in \Cref{ex.setup.blp}.
Away from the null, the residual term also accounts for estimation of the prediction variance.

For each refit, define
\begin{align*}
(\widehat\alpha_{k,b},\betahat_{k,b})
&=
\argmin_{a,t}\sum_{i\in\stilde_k}
\bc{Y_i-a-t\etahat_{k,b}(X_i)}^2,\\
\widehat z_{k,b,i}
&=
\etahat_{k,b}(X_i)
-
\frac1{2m}\sum_{j\in\stilde_k}\etahat_{k,b}(X_j),\\
\widehat e_{k,b,i}
&=
Y_i-\widehat\alpha_{k,b}
-\betahat_{k,b}\etahat_{k,b}(X_i),\\
\widehat\phi_{k,b,i}
&=
\frac{
\widehat z_{k,b,i}\widehat e_{k,b,i}
}{
(2m)^{-1}\sum_{j\in\stilde_k}\widehat z_{k,b,j}^2
}.
\end{align*}
Use $\widehat\phi_{k,b,i}$ in place of $f_{\etahat_{k,b}}(W_i)$ in \cref{eq.inference.pooled_scale,eq.inference.contrast}.
Its common positive denominator cancels from $Z_{k,b}$, so the implementation may equivalently use $\widehat z_{k,b,i}\widehat e_{k,b,i}$.
These pooled quantities do not depend on the test fold's primary/reference labels.

The correlation result extends provided
$$
\frac{
\sum_{i\in\stilde_k}
\bc{\widehat\phi_{k,1,i}-\phi_{\etahat_{k,1}}(W_i)}^2
}{
2m\,P\phi_{\etahat_{k,1}}^2
}
\pto0,
$$
with the theorem's normalized-score assumptions applied to $\phi_\eta$.
Indeed, with $Z(x)$ and $v(x)$ as in the proof of \Cref{thm.inference} in \Cref{app.proofs.inference}, for two pooled score vectors $x,y$, the balanced contrast satisfies
$$
\E{\bc{Z(x)-Z(y)}^2}
=
2\norm{v(x)-v(y)}_2^2.
$$
Relative score consistency makes the normalized vectors converge in probability.
Their squared distance is bounded by four, giving convergence in mean square of the contrasts and leaving the correlation proof unchanged.

For coverage using the original OLS estimators, also require
$$
\betahat_k-\theta_{\etahat_k}
=
\frac2n\sum_{i\in\stilde_k}\phi_{\etahat_k}(W_i)
+
o_p\of{\sqrt{\frac{P\phi_{\etahat_k}^2}{n}}},
$$
and relative consistency of their marginal standard errors.
In this application, the scale-balance condition concerns the slope influence scales $\sqrt{P\phi_{\etahat_k}^2}$.

\end{appendices}

\end{document}